\documentclass{article}

\usepackage{multirow}
\usepackage{amssymb}
\usepackage{amsmath}
\usepackage{amsthm}
\usepackage{thm-restate}
\usepackage[dvipsnames]{xcolor}
\usepackage{comment}

\usepackage[pdftex, plainpages = false, pdfpagelabels, 
                 bookmarks=false,
                 bookmarksopen = true,
                 bookmarksnumbered = true,
                 breaklinks = true,
                 linktocpage,
                 pagebackref,
                 colorlinks = true,  
                 linkcolor = blue,
                 urlcolor  = blue,
                 citecolor = red,
                 anchorcolor = green,
                 hyperindex = true,
                 hypertexnames = false,
                 hyperfigures
                 ]{hyperref}
\usepackage{complexity} 
\usepackage{enumitem}
\setlist{nosep}
\usepackage{cleveref}
\usepackage{subcaption}
\usepackage{dirtytalk}

\usepackage{tikz}

\usetikzlibrary{arrows,shapes,positioning,shadows,trees}

\usepackage{vmargin}
\setmarginsrb{1in}{1in}{1in}{1in}{0pt}{0pt}{0pt}{6mm}

\usepackage{algorithmic}
\usepackage{algorithm}
\usepackage[utf8]{inputenc}
\usepackage[LGR,T1]{fontenc} 
\usepackage[T1]{fontenc}
\usepackage{amsmath}

\newcommand{\yes}{\textsc{Yes}}

\newtheorem{theorem}{Theorem}
\newtheorem{lemma}{Lemma}

\newtheorem{claim}{Claim}[section]
\newtheorem{observation}[lemma]{Observation}
\newtheorem{proposition}{Proposition}[section]
\newtheorem{reduction rule}{Reduction Rule}[section]
\newtheorem{marking-scheme}{Marking Scheme}[section]
\newtheorem{definition}{Definition}[section]

\newenvironment{claimproof}[1]{\par\noindent\textit{Proof of Claim.}\space#1}{\hfill $\diamond $\medskip}

\newcommand{\defparproblem}[4]{
    \vspace{3mm}
    \noindent\fbox{
        \begin{minipage}{.95\textwidth}
            \begin{tabular*}{\textwidth}{@{\extracolsep{\fill}}lr} \textsc{#1}  & {\bf{Parameter:}} #3 \\ \end{tabular*}
            {\bf{Input:}} #2  \\
            {\bf{Question:}} #4
        \end{minipage}
    }
    \vspace{2mm}
}

\usepackage{xspace}
\usepackage{multicol}
\newcommand{\RNum}[1]{\uppercase\expandafter{\romannumeral #1\relax}}
\usepackage{soul}
\usepackage{lineno}

\newcommand{\cmso}{$\textsc{CMSO}_2$\xspace}

\newcommand{\tail}{\mathrm{tl}}
\newcommand{\true}{\mathrm{true}}
\newcommand{\head}{\mathrm{hd}}
\newcommand{\type}{\mathbf{type}}
\newcommand{\Unb}{\textsf{Solve-Unbr-ALG}\xspace}

\newcommand{\Und}{\textsf{Understand-ALG}\xspace}
\newcommand{\UndU}{\textsf{Understand-Unbr-ALG}\xspace}

\DeclareMathOperator{\torso}{torso}
\DeclareMathOperator{\tw}{tw}
\newcommand{\deoct}{\textsc{DEOCT}\xspace}
\newcommand{\deoctfull}{\textsc{Directed Edge Odd Cycle Transversal}\xspace}
\newcommand{\dmc}{\textsc{DMC}\xspace}
\newcommand{\dmcfull}{\textsc{Directed Multicut}\xspace}
\newcommand{\demcu}{\textsc{DEMCU}\xspace}
\newcommand{\demcufull}{\textsc{Directed Edge Multiway Cut-Uncut}\xspace}
\newcommand{\dsc}{\textsc{DSC}\xspace}
\newcommand{\dscfull}{\textsc{Directed Steiner Cut}\xspace}
\newcommand{\minecmso}[1]{{\sc min-Edge-$\textsc{CMSO}_2$[#1]}}
\newcommand{\maxecmso}[1]{{\sc max-Edge-$\textsc{CMSO}_2$[#1]}}
\newcommand{\eqecmso}[1]{{\sc eq-Edge-$\textsc{CMSO}_2$[#1]}}
\newcommand{\xcmso}[1]{{\sc x-$\textsc{CMSO}_2$[#1]}}
\newcommand{\cmsop}[1]{{\sc $\textsc{CMSO}_2$[#1]}}

\usepackage{tikz}
\usetikzlibrary{arrows}
\usetikzlibrary{arrows.meta}
\usetikzlibrary{decorations.pathmorphing, decorations.pathreplacing, decorations.shapes}
\usetikzlibrary{decorations.markings}
\usetikzlibrary{calc}
\usetikzlibrary{shapes.geometric}

\tikzset{
  circ/.style = {circle,draw,fill,inner sep=1.1pt},
  circg/.style = {circle,draw=lightGray,fill=lightGray,inner sep=1.1pt},
  invisible/.style = {circle,draw=none,inner sep=0pt,font=\tiny},
  nonedge/.style={decorate,decoration={snake,amplitude=.3mm,segment length=1mm},draw}
}

\usepackage{authblk}

\title{Tractability of Hard Digraph Cut Problems on Eulerian Instances: Treewidth Reduction and Reducing CMSO to Unbreakable Eulerian Instances}
\date{}

\author[1]{Esther Galby%
}

\author[2]{Roohani Sharma%
\thanks{Supported by the Young Scientist Fellowship of the Institute for Basic Science (IBS-R029-Y8).}}

\author[1]{Bart van der Steenhoven%
}

\affil[1]{Chalmers University of Technology and University of Gothenburg, Sweden\\
\texttt{galby@chalmers.se, bartj@chalmers.se}}

\affil[2]{DIMAG, Institute for Basic Science, South Korea\\
\texttt{roohani@ibs.re.kr}}

\begin{document}

\maketitle

\begin{abstract}
Eulerian digraphs retain the directional structure of general digraphs while exhibiting a symmetry of cuts reminiscent of undirected graphs. We show that this symmetry extends beyond individual cuts to two fundamental techniques in parameterized algorithms for undirected cut problems: {\em treewidth reduction} and {\em recursive understanding and randomized contractions}.

Both the mentioned techniques have proved very useful for dealing with undirected graph cut problems as each of them reduces the problem at hand on general undirected instances to ones that have bounded treewidth or that are highly connected in a way, respectively.

In this work we prove an
edge-cut analogue of the treewidth-reduction theorem of Marx, O'Sullivan, and Razgon [ACM TALG 2013],
and an Eulerian analogue of the reduction of CMSO$_2$ model checking to unbreakable instances due to Lokshtanov, Ramanujan, Saurabh, and Zehavi [ICALP 2018]. Both these results now allow us to reduce several problems on general Eulerian digraphs to Eulerian digraphs whose underlying undirected graph has bounded treewidth, or to Eulerian instances that are ``highly connected'' in the directed (and undirected) sense.
Together, our results bring to Eulerian digraphs two fundamental techniques from the undirected parameterized cut toolkit for which no counterparts are known on arbitrary digraphs.

Beyond their structural significance, these results
have substantial algorithmic consequences: they yield
fixed-parameter algorithms for several directed cut
problems that are W[1]-hard on general digraphs.
For example, we show that {\sc Directed Multicut} and {\sc Directed Edge Odd Cycle Transversal} are
fixed-parameter tractable on Eulerian digraphs when
parameterized by the solution size and the number of
terminal pairs, and just the solution size respectively. This stands in sharp contrast to general
digraphs, where the first problem is W[1]-hard even for four
terminal pairs, and the second problem is W[1]-hard in general.
We also obtain fixed-parameter algorithms
for 
{\sc Directed Edge Multiway Cut-Uncut} and
{\sc Directed Steiner Cut}.

\end{abstract}

\section{Introduction}\label{sec:intro}

Cut problems are central to algorithmic graph theory.  Given a graph and a connectivity requirement, the goal is to delete a minimum number of vertices or edges to destroy prescribed paths or cycles. Classical examples include {\sc Minimum $s$-$t$ Cut}, {\sc Multiway Cut}, {\sc Multicut}, {\sc Feedback Vertex Set} and {\sc Odd Cycle Transversal}.
Over the past two decades, the study of NP-hard cut problems, in the realm of parameterized complexity, has produced a collection of powerful structural and algorithmic techniques---including important separators~\cite{DBLP:conf/wg/Marx11}, treewidth reduction~\cite{DBLP:journals/talg/MarxOR13}, shadow removal~\cite{DBLP:conf/stoc/MarxR11}, recursive understanding and randomized contractions~\cite{DBLP:conf/stoc/GroheKMW11,DBLP:journals/siamcomp/ChitnisCHPP16}, and flow augmentation~\cite{DBLP:journals/jacm/KimKPW25,DBLP:journals/talg/KimKPW24}---that now form much of the toolkit of parameterized graph cut algorithms. 
While these techniques have led to a remarkably ``complete'' understanding of arguably the fundamental NP-hard graph cut problems on undirected graphs, 
the directed setting remains significantly more challenging. Only some of these techniques extend to directed graphs, and several fundamental problems that are fixed-parameter tractable (FPT) on undirected graphs become W[1]-hard on directed graphs. 
Despite their hardness on general digraphs, the parameterized complexity of these problems remains poorly understood even on natural restricted classes of digraphs, like planar digraphs.

A particularly natural and fundamental class of directed graphs is that of \emph{Eulerian digraphs}. 
A digraph is Eulerian if it admits a closed directed trail traversing every arc exactly once. Rooted in the classical study of traversability of networks, this notion arises naturally in routing and network inspection problems,
where one seeks to traverse every connection without repetition.
Remarkably, this global traversal property admits a simple local characterization: a digraph whose underlying undirected graph is connected is Eulerian if and only if every vertex has equal indegree and outdegree.
This balance condition captures the principle of flow conservation underlying the theory of network flows and circulations. 
Despite this seemingly simple restriction,
Eulerian digraphs can exhibit arbitrarily complex underlying graph structure while retaining the asymmetry characteristic of directed graphs.
They thus provide a rich setting in which to investigate how a local balance condition can influence the global structure and computational complexity of directed graphs.

Eulerian digraphs have been the subject of repeated studies from a structural perspective \cite{dean1986smallest, van1987circuits, johnson2002eulerian, huang2013large, DBLP:conf/icalp/CavallaroKK26}.
However, from an algorithmic perspective the balance property seems to make it so that Eulerian digraphs 
form a natural intermediate class between undirected and directed graphs. 
Eulerian digraphs have repeatedly
emerged as the right setting in which algorithmic results non-trivially interpolate between the undirected and directed settings, see for example \cite{ frank1988connectivity, gutin2014parameterized, DBLP:conf/stoc/CavallaroKK24} or \cite{DBLP:books/sp/18/Wahlstrom18} for a survey. This raises a more fundamental question: 

{\em Can the structure that is provided by Eulerian digraphs restore the algorithmic machinery that is missing for directed cut problems?}

We answer this question affirmatively by showing that two fundamental techniques from the undirected parameterized cut toolkit can be recovered in the Eulerian setting.
The first, \emph{treewidth reduction} of Marx, O'Sullivan and Razgon~\cite{DBLP:journals/talg/MarxOR13}, says that the vertices of all minimal $s$-$t$ separators of size at most $k$ in an undirected graph can be captured in a graph of bounded treewidth, where the bound depends only on $k$. 
Second, the meta-theorem of Lokshtanov, Ramanujan, Saurabh and Zehavi~\cite{DBLP:conf/icalp/LokshtanovR0Z18} that captures \emph{recursive understanding},
reduces the study of arbitrary instances of \cmso definable problems, to ones that are ``highly connected'',
or \emph{unbreakable}.
In this paper, we establish both these (meta) theorems for Eulerian digraphs. 
We show an edge-cut analogue of
treewidth reduction that can be exploited on Eulerian digraphs,
and an Eulerian analogue of the reduction of \cmso model
checking to unbreakable instances.
We explain both techniques and state our structural results
below.
Our contribution of this work goes beyond establishing tractability for individual problems: these structural results
provide general algorithmic tools that yield FPT algorithms on Eulerian digraphs for several directed cut problems that are W[1]-hard on
general digraphs.

\subparagraph{Treewidth reduction.}
One of the earliest structural successes of parameterized complexity for cut problems was the study of the space of all small minimal $s$-$t$-separators. The understanding of this structure is important,
as in many NP-hard generalizations of {\sc Minimum $s$-$t$ Cut}, 
for example {\sc Multiway Cut},
the solution is a non-trivial union of some overlapping minimal $s$-$t$ cuts.
In a seminal work, Marx, O'Sullivan and Razgon~\cite{DBLP:journals/talg/MarxOR13} showed that all vertices participating in a minimal $s$-$t$
 separator of size at most $k$
 in an {\em undirected} graph can be captured in a graph of treewidth $2^{k^{O(1)}}$. Specifically, both the graph induced by these vertices and the graph obtained by bypassing all other vertices have treewidth bounded by a function of $k$ alone. This result enables cut problems whose optimal solutions are a (non-trivial) combination of minimal $s$-$t$ separators, for possibly different choices of $s$ and $t$, to be reduced to bounded-treewidth instances, which are then solvable using meta-theorems such as Courcelle's theorem~\cite{DBLP:journals/ita/Courcelle92} for a large class of problems.

No directed counterpart of this treewidth reduction for digraphs is known. 
Directed minimal $s$-$t$ separators lack the structure of undirected separators, especially since undirected treewidth only cares about small separations in the underlying undirected graph. And trying to bound the directed treewidth for the digraph cases, usually is not as useful as bounding undirected treewidth as far as algorithmic applications are concerned.
In fact, Marx, O'Sullivan and Razgon~\cite{DBLP:journals/talg/MarxOR13} explicitly ask if
{\em there is some way of extending these results (at least partially) to directed graphs?}
This question has remained open ever since. Our first result concerns precisely this; we prove an edge-cut version of treewidth reduction and show its utility for cut problems on Eulerian digraphs. We refer to \Cref{sec:twrproofs} for missing definitions.

\begin{restatable}{theorem}{thmtwreduction}\label{thm:twreduction}
    Let $G$ be a graph, $T \subseteq V(G)$ and $k$ an integer. Let $I$ be the set of all edges of $G$ that are part of a minimal $s-t$ cut of size at most $k$ for some $s,t \in T$. There is a $f(k, |T|) \cdot (|E(G)| + |V(G)|)$ time algorithm 
    that computes a set $W \subseteq V(G)$ such that $W$ contains $T$ and all endpoints of edges in $I$, and $\tw(\torso(G,W)) \leq h(k,|T|)$ for some functions $f(k, |T|)$ and $h(k, |T|)$ that both are  $|T|^2 2^{O(k^2)}$.
\end{restatable} 

\subparagraph{Reduction to unbreakable instances.}
Call a graph $(q,k)$-\emph{unbreakable} if one cannot separate two sets of $q$ vertices from each other by deleting at most $k$ vertices. Every small separation of an unbreakable graph is lopsided, shaving off only a bounded piece, so a solution of size at most $k$ cannot cut the graph in two, and its effect is confined to a bounded region that can be located by brute force or by colour coding (randomly coloring small parts that suggest some part of the solution). This suggests a win/win strategy, 
and it is the principle of \emph{recursive understanding}, which emerged in the work of Grohe, Kawarabayashi, Marx and Wollan~\cite{DBLP:conf/stoc/GroheKMW11} and of Kawarabayashi and Thorup~\cite{DBLP:conf/focs/KawarabayashiT11}. 
Chitnis, Cygan, Hajiaghayi, Pilipczuk and Pilipczuk~\cite{DBLP:journals/siamcomp/ChitnisCHPP16} complemented this with a colour-coding infrastructure, \emph{randomized contractions}, that performs the understanding step for a range of cut problems, and thereby obtained fixed-parameter algorithms for {\sc Unique Label Cover}, {\sc Steiner Cut} and {\sc Node Multiway Cut-Uncut}. Lokshtanov, Ramanujan, Saurabh and Zehavi~\cite{DBLP:conf/icalp/LokshtanovR0Z18} observed that a large part of this approach could be automated for all problems in {\sc CMSO}$_2$.

\begin{theorem}[\cite{DBLP:conf/icalp/LokshtanovR0Z18}]\label{thm:lrsz}
Let $\phi$ be a {\sc CMSO}$_2$ sentence. For every $k \in \mathbb{N}$ there exists $q \in \mathbb{N}$ such that, if $\phi$ can be tested on $(q,k)$-unbreakable graphs in time $O(n^d)$ for some $d > 4$, then $\phi$ can be tested on general graphs in time $O(n^d)$.
\end{theorem}

Theorem~\ref{thm:lrsz} essentially says that if designing FPT algorithms for \cmso definable problems, it is enough to design such algorithms restricted to instances that are unbreakable.

Unfortunately, an analogue of Theorem~\ref{thm:lrsz} does not hold for directed graphs. 
Unbreakability is a statement about \emph{undirected} separations, and it is useful only because a small solution of an undirected cut problem \emph{is} a small separation, so that unbreakability constrains it. 
In digraphs a set of $\ell$ arcs may separate $s$ {\em from} $t$ (and not $t$ from $s$ at the same time) 
while the corresponding cut of the underlying undirected graph has $\ell + n$ edges, the $n$ arcs in the reverse direction being irrelevant for a solution cut. 
Since the arc sets one wishes to delete are not small cuts of the underlying undirected graph, the win/win step loses its base case before the recursion is ever reached.

For our second main result we prove a directed Eulerian analogue of the result by Lokshtanov, Ramanujan, Saurabh and Zehavi~\cite{DBLP:conf/icalp/LokshtanovR0Z18}. We refer to \Cref{sec:reduce-to-unbreakable} for missing definitions.

\begin{restatable}{theorem}{thmcmsofpt} 
\label{thm:cmsofpt}
    Let $\textsc{x} \in \{ \textsc{min-Edge, max-Edge, eq-Edge}  \}$, let $\tilde\psi$ be a 
    \cmso
    sentence and let $\tilde c: \mathbb{N} \to \mathbb{N}$. There exists a $\tilde s : \mathbb{N} \to \mathbb{N}$ such that if 
    \xcmso{$\tilde\psi$} 
    parameterized by $k$ is FPT on $(\tilde s(k),\tilde c(k))$-edge-unbreakable Eulerian structures, then \xcmso{$\tilde\psi$} parameterized by $k$ is (nonuniformly) FPT on general Eulerian structures.
\end{restatable}

Interestingly, even though (induced) subgraphs of an Eulerian digraph, may not be Eulerian, our meta-theorem for Eulerian digraphs reduce the Eulerian instances to {\em Eulerian unbreakable instances}. This allows one to exploit {\em both Eulerianess and unbreakability} simultaneously of the reduced instance.
Also observe that this has a considerably broader scope than any individual application: once an algorithm is available for arc-unbreakable Eulerian digraphs, it automatically lifts to general Eulerian digraphs.

\subparagraph{Why Eulerianness restores both techniques.}  The common obstruction behind treewidth reduction and reduction to unbreakability is that both rely on small undirected separations. 
In a general digraph, a directed cut need not correspond to a small cut in the underlying graph, so neither framework has any leverage. Eulerianness restores precisely this connection. Every directed cut of size $\ell$ is accompanied by an undirected cut of size exactly $2\ell$, allowing the undirected structure theory to become relevant once again. This observation underlies both of our structural results. To the best of our knowledge, treewidth reduction and recursive understanding (or, equivalently, reduction to unbreakable instances) are the only two major techniques in the parameterized cut toolbox that had resisted extension to directed graphs. Our main contribution is to show that both admit natural Eulerian analogues.

The above connection also explains why our results focus on arc-deletion problems. The balance condition relates the numbers of arcs crossing a partition, but provides no analogous guarantee for vertex separators.
Moreover, the standard transformations between arc- and vertex-deletion variants do not, in general, preserve Eulerianness, and vertex-deletion variants can remain as hard on Eulerian digraphs as on general digraphs as explained in~\cite{DBLP:books/sp/18/Wahlstrom18}.

\subsection{Applications} 
We illustrate the power of these structural theorems on several classical cut problems.
Our first application is \textsc{Directed Multicut}. 

\defparproblem{\textsc{Directed Multicut (DMC)}}{A directed graph $D$, a set $\mathcal{T} \subseteq V(G) \times V(G)$ of terminal pairs, and a positive integer $k$}{$k + |\mathcal{T}|$}{Does there exist $S \subseteq E(D)$ of size at most $k$, such that there is no $s\rightarrow t$ path in $D-S$ for any $(s,t) \in \mathcal{T}$?
}

On general digraphs the problem is W[1]-hard already for four terminal pairs~\cite{DBLP:journals/toct/PilipczukW18a}, whereas on DAGs it becomes fixed-parameter tractable when the number of terminal pairs is included in the parameter (and remains hard with just the solution size as the parameter)~\cite{DBLP:journals/siamdm/KratschPPW15}. 
We show that the same parameterization, as DAGs, leads to an FPT algorithm on Eulerian digraphs.

\begin{restatable}{theorem}{dmcthm}\label{thm:dmc}
    {\textsc{Directed Multicut}} on Eulerian digraphs can be solved in $f(k + |\mathcal{T}|) \cdot (|V(D)| +|E(D)|)^{O(1)}$ time, where $f$ is a computable function of $k$ and $\mathcal{T}$.
\end{restatable}


Our treewidth-reduction theorem also yields a fixed-parameter algorithm for \textsc{Directed Edge Odd Cycle Transversal}, which is W[1]-hard in general~\cite{DBLP:conf/soda/LokshtanovR0Z20}.

\defparproblem{\textsc{Directed Edge Odd Cycle Transversal (DEOCT)}}{A directed graph $D$ and a positive integer $k$}{$k$}{Does there exist $S \subseteq E(D)$ of size at most $k$, such that there is no cycle of odd length in $D-S$?
}

\begin{restatable}{theorem}{deoctthm}\label{thm:deoct}
    {\textsc{Directed Edge Odd Cycle Transversal}} can be solved in $2^{2^{O(k^2)}} \cdot n^{O(1)}$ time on Eulerian digraphs.
\end{restatable}


Finally, our \cmso reduction to arc-unbreakable instances yields nonuniform fixed-parameter algorithms for \textsc{Directed Edge Multiway Cut-Uncut} and \dscfull. To contrast this, we prove that both problems are W[1]-hard on general digraphs. Both problems are known to be FPT on undirected graphs \cite{DBLP:journals/siamcomp/ChitnisCHPP16}.

\defparproblem{\textsc{Directed Edge Multiway Cut-Uncut (DEMCU)}}{A directed graph $D$,  a set $T \subseteq V(D)$ of terminals, an equivalence relation $R$ on $T$ and a positive integer $k$}{k}{Does there exist $S \subseteq E(D)$ of size at most $k$ such that for all $t, t' \in T$, the terminals $t$ and $t'$ are in the same strongly connected component of $D-S$ if and only if $(t,t') \in R$?}

\begin{restatable}{theorem}{thmdemwcu}\label{thm:demwcu}
    \textsc{Directed Edge Multiway Cut-Uncut} on Eulerian digraphs is (nonuniformly) FPT parameterized by $k$.
\end{restatable}

\defparproblem{\textsc{Directed Steiner Cut (DSC)}}{A directed graph $D$, a set $T \subseteq V(G)$ of terminal pairs, and positive integers $k$ and $p$}{$k+p$}{Does there exist $S \subseteq E(D)$ of size at most $k$, such that at least $p$ strongly connected components of $D-X$ contain at least one terminal vertex?
}

\begin{restatable}{theorem}{thmdsc}\label{thm:dsc}
    \dscfull on Eulerian digraphs is (nonuniformly) FPT parameterized by $k+p$. 
\end{restatable}

\begin{theorem}
    \demcufull parameterized by $k$ and \dscfull parameterized by $k+p$ are W[1]-hard on general digraphs.
\end{theorem}

The applications presented above are unlikely to exhaust the reach of the framework. In particular, we expect the reduction to arc-unbreakable instances to provide a useful starting point for other cut problems on Eulerian digraphs.
Two important remaining challenges with Eulerian digraphs are to determine the parameterized complexity of \dmcfull and \textsc{Symmetric Directed Multicut} (see \cite{DBLP:conf/iwpec/EibenRW22}) when parameterized only by the solution size.

\section{Treewidth reduction for Eulerian Digraphs}\label{sec:twr}

The properties of Eulerian digraphs allow us to relate minimal directed \emph{arc-cuts} to minimal undirected \emph{edge-cuts}, not vertex-separators.
Thus, to be able to apply the treewidth reduction technique to cut problems in Eulerian digraphs, we need to modify the original theorem to preserve undirected minimal terminal \emph{cuts} of size at most $k$, as seen in \Cref{thm:twreduction}.

\thmtwreduction*
 
More specifically, our treewidth-reduction proof follows the framework of Marx, O'Sullivan and Razgon~\cite{DBLP:journals/talg/MarxOR13}, adapted to edge cuts. 
The proof works by finding, for each pair $(s,t)$ of terminals in $T$, a vertex set $W_{s,t}$ containing the endpoints of all minimal $s-t$ cuts of size at most $k$, yet of which the torso has bounded treewidth. These sets can be combined to form $W$, which by properties of the torso operation also has bounded treewidth. The set $W_{s,t}$ is found by induction on the \emph{excess} of the cuts that we are looking for. The excess of a cut is equal to the amount it is larger than a minimum cut. Intuitively, for excess 0 (i.e. minimum cuts) one can use flow techniques to find a vertex set to cover the cuts. For excess $i$, one can observe that these cuts must contain edges inside regions sitting between minimum cuts, and we can construct auxiliary graphs where such edges now appear in cuts of excess $i-1$, hence allowing the inductive argument to take effect. The formal proof can be found in \Cref{sec:twrproofs}.

We continue to demonstrate how this tool can be used to prove powerful algorithmic results for cut problems on Eulerian digraphs. As a warming-up, we show that a simple application leads to an FPT algorithm for \dmcfull parameterized by $k$ and the number of terminals. Next, we show that a more technically involved application leads to an FPT algorithm for \deoctfull parameterized by only $k$. Both of these results are in stark contrast to the complexity landscape for general directed graphs, where both problems are W[1]-hard for the same parameterizations~\cite{DBLP:journals/toct/PilipczukW18a,DBLP:conf/soda/LokshtanovR0Z20,DBLP:conf/innovations/LochetLM0SZ20}. 

\subsection{Directed Multicut}
Before we begin, we need to introduce one more definition. The treewidth reduction technique relies on a torso operation to shrink the graph to a specified region containing all minimal terminal cuts of a certain size, while simultaneously ensuring that connectivity outside of this region is preserved. To replicate this for out setting, we introduce the notion of a \emph{directed torso}. 
\begin{definition}
 Let $D$ be a directed graph, $k$ a natural number and $U \subseteq V(D)$, then the \emph{(directed) $k$-torso} of $U$ in $D$, denoted by $\torso(D,U, k)$ is the graph $D[U]$ where for each $a,b \in U$ we additionally add $k+1$ arcs $(a,b)$ if $D$ has an $a$ to $b$ path internally disjoint from $D[U]$.
\end{definition}
In our applications we want to argue that a solution for a cut problem in the directed torso can be transferred back to a solution of the original digraph. Adding $k+1$ copies of all arcs introduced by the torso operation ensures that a minimal solution of size at most $k$ only contains arcs in the induced digraph $D[U]$, which simplifies the required argument. We are now ready to prove the first main algorithmic result.

\dmcthm*

\begin{proof}
    Let $C \subseteq E(D)$ be an optimal solution to the \textsc{Directed Multicut} instance $(D,\mathcal{T},k)$. Let $T$ be set of all terminals in $\mathcal{T}$. 
    For every terminal pair $(s,t) \in \mathcal{T}$, the solution $C$ must contain a minimal $s\rightarrow t$ cut of $D$. 
    As $S$ is optimal, every arc $e \in C$ is part of a minimal $s \rightarrow t$ cut for some terminal pair $(s,t) \in \mathcal{T}$. Furthermore, there is such a minimal terminal-pair cut that contains $e$ and is a subset of $C$. 
    Additionally, since $D$ is an Eulerian digraph, each minimal $s \rightarrow t$ cut of size $\ell$ corresponds to a minimal $s - t$ cut of size $2\ell$ in the underlying undirected graph $G$ of $D$. 
    We apply \Cref{thm:twreduction} to $G$ with cut size $2k$ and terminal set $T$ to obtain a set $W \subseteq V(G)$ as described in the theorem. We prove that is suffices to solve the \textsc{Directed Multicut} problem on the directed $k$-torso of $W$ in $D$. 
    \begin{claim}
        $(D,\mathcal{T},k)$ is a YES-instance for \textsc{Directed Multicut} if and only if $(\torso (D,W,k),\allowbreak \mathcal{T}, k)$ is a YES-instance. 
    \end{claim}
    \begin{claimproof}
        Let $C$ be a minimum arc set of size at most $k$ in $D$ such that $D-C$ has no $s\rightarrow t$ path for any $(s,t) \in \mathcal{T}$. By the argument preceding the claim, all arcs of $C$ are part of some minimal terminal-pair cut of size at most $2k$ in the undirected graph $G$. Hence all arcs of $C$ are in $D[W]$, and thus present in $\torso(D,W,k)$. We argue that $C$ also cuts all terminal pairs in the torso graph. 
        
        For sake of a contradiction, assume that there is an $s \rightarrow t$ path $P$ in $\torso(D,W, k) - C$ for some $(s,t) \in \mathcal{T}$. This implies there is also an $s \rightarrow t$ path in $D - C$, as any arc $e$ on $P$ that is not in $D$ can be replaced by a directed path from the tail of $e$ to its head, internally disjoint from $D[W]$, the existence of which is guaranteed by the presence of $e$ in $\torso(D,W,k)$. This path is not hit by $C$ as $C$ contains only arc with both endpoints in $W$. This contradicts the fact that $C$ hits all $s \rightarrow t$ paths in $D$. 

        Conversely, let $C$ be a minimum directed multicut of size at most $k$ in $\torso(D,W, k)$. As we use a $k$-torso, we know that a minimum $C$ cannot contain any arcs not in $D[W]$. 
        We claim that $C$ also hits all terminal-pair paths in $D$. For a contradiction, assume there is a path $P$ from $s$ to $t$ in $D-C$ for a $(s,t) \in \mathcal{T}$. If $P$ has a subpath internally vertex-disjoint from $W$, then we can replace this part of the path by one of the arcs of $\torso(D,W,k)$ that this subpath gives rise to. This arc cannot be in $C$, and hence we find a path from $s$ to $t$ in $\torso(D,W,k) - C$, a contradiction.   
    \end{claimproof}
    
    Thus we can focus our attention to the instance $(\torso(D,W,k), \mathcal{T},k)$. Observe that $\torso(D,W,k)$ need not be an Eulerian digraph. 
    By \Cref{thm:twreduction} we know that, for the underlying undirected graph $G$ of $D$, the graph $\torso(G,W)$ has treewidth at most $h(2k,|T|)$. When we disregard parallel edges (which do not affect the treewidth regardless), the undirected graph underlying $\torso(D,W,k)$ is a subgraph of $\torso(G,W)$ and thus must also have the same upperbound on the treewidth. As \textsc{Directed Multicut} can be described with a {\sc CMSO}$_2$ formula of size dependent only on $|\mathcal{T}|$, we can use Courcelle's theorem to solve the instance $(\torso(D,W, k), \mathcal{T},k)$ in time $f(k,|T|) \cdot n$ for some function $f$. 
\end{proof}

\subsection{Directed Edge Odd Cycle Transversal}
\label{ssec:deoct}

We now turn to a less direct application of treewidth reduction. Unlike the earlier application, it is not
immediately clear how to exploit treewidth reduction for
{\sc Directed Edge Odd Cycle Transversal}. The difficulty
is that the requirement of destroying all odd {\em directed cycles} does not directly translate into finding a collection of small minimal $s$-$t$ cuts, whereas the classical FPT algorithm for undirected {\sc Odd Cycle Transversal} exploits iterative compression to reduce the problem to finding minimum $s$-$t$ separators~\cite{DBLP:journals/orl/ReedSV04}.

A natural approach is therefore to express the requirement of destroying all odd directed cycles as a collection of directed reachability constraints, as in
{\sc Directed Multicut}. To this end, we first use the FPT $2$-approximation algorithm of Lokshtanov, Ramanujan, Saurabh, and Zehavi~\cite{DBLP:conf/soda/LokshtanovR0Z20} to find a set $M$ of at most $2k$ vertices intersecting every odd directed cycle, or conclude that no solution of size at most $k$ exists. We then construct an auxiliary bipartite digraph by taking two copies of each vertex and replacing every arc by two arcs connecting opposite copies. The resulting digraph remains Eulerian, and odd directed cycles can be detected through reachability between the two copies of vertices in $M$. This suggests a {\sc Directed Multicut} formulation with at most $4k$ terminal pairs, in which the two copies of each vertex in $M$ must be separated in both directions.

However, this formulation does not yield a reduction to ordinary {\sc Directed Multicut}. Each arc of the original digraph corresponds to a pair of arcs in the auxiliary digraph, and deleting an original arc requires deleting both members of the pair. An arbitrary multicut need not respect this pairing, and hence need not correspond to a solution of the original problem. We are thus led to a constrained variant of {\sc Directed Multicut} in which the two arcs of each pair must either both be deleted or both be retained.

The key observation is that this additional constraint does not prevent us from exploiting the structure of small minimal cuts. Let $\phi$ denote the involution exchanging the two copies of each vertex, and hence the two arcs corresponding to each original arc. We show that every inclusion-minimal multicut $S$ respecting the pairing constraint can be expressed as $S=A\cup\phi(A)$, where $A$ is an inclusion-minimal ordinary multicut. Consequently, every arc of $S$ belongs to a minimal cut of size at most $2k$ between some prescribed terminal pair. Since the auxiliary digraph is Eulerian, these directed cuts correspond to undirected cuts of size at most $4k$. This allows us to apply our treewidth-reduction theorem (Theorem~\ref{thm:twreduction}) to obtain an equivalent instance of bounded treewidth while preserving the pairing constraint.

Although treewidth reduction yields an equivalent instance of bounded treewidth, enforcing the pairing constraint introduces an additional difficulty. In an arbitrary tree decomposition, the two arcs of a pair may be handled in different parts of the decomposition, requiring the dynamic programming algorithm to maintain consistency between their deletion decisions.

We overcome this difficulty by proving that the reduced graph admits a bounded-width tree decomposition in which every bag is closed under the involution $\phi$. To establish this, we consider the graph obtained by identifying each vertex with its counterpart under $\phi$. We show that this graph must have bounded treewidth otherwise, the excluded-grid theorem~\cite{DBLP:journals/jctb/ChuzhoyT21}, together with a grid-minor lifting argument that exploits the involution, would imply the existence of an arbitrarily large grid minor in the original reduced graph, contradicting its bounded treewidth. A tree decomposition of the identified graph can then be lifted to one of the original graph by replacing each vertex with its two counterparts. The resulting decomposition allows us to enforce the pairing constraint locally, and dynamic programming over it yields the desired FPT algorithm.

From this point onwards, we will say that the treewidth of a directed graph $D$, written as $\tw(D)$, is equal to treewidth of the undirected underlying graph of $D$. One should not confuse this with the directed treewidth of $D$, which is a different concept not relevant for this work. 

\deoctthm* 

    Let $(D,k)$ be the input instance of \deoct. 
    As a first step, we can use the FPT 2-approximation algorithm by Lokshtanov, Ramanujan, Saurabh and Zehavi~\cite{DBLP:conf/soda/LokshtanovR0Z20}. This will in time $2^{O(k^2)} \cdot n^{O(1)}$ either report that $D$ has no vertex set (and thus no arc set) of size at most $k$ that hits all odd cycles, or it find a vertex set $M$ of size at most $2k$ such that $D-M$ is odd-cycle-free. In the former case we immediately return that this is a {\sc NO}-instance.
    
    Next, we reduce the \deoct instance $(D,k)$ to a \dmcfull instance $(D', \mathcal{T}, 2k)$ defined as follows. The graph $D'$ has as vertex set $\{v^0, v^1 \mid v \in V(D)\}$ and has arcs $(u^0, v^1)$ and $(u^1, v^0)$ for each arc $(u,v)$ of $D$. Finally, $\mathcal{T} = \{ (t^0, t^1), (t^1, t^0) \mid t \in M\}$ and hence $|\mathcal{T}| \leq 2|M| \leq 4k$. We define an involution $\phi$ on $V(D') \cup E(D')$ with $\phi (v^i) = v^{1-i}$ for all $v^i \in V(D')$ and $\phi((u, v)) = (\phi(u), \phi (v))$.
    We say that a set $S \subseteq V(D) \cup  E(D)$ is \emph{symmetric} when $\phi(S) := \{\phi(e) \mid e \in S\}$ is equal to $S$. 
    The reduction relates the two instances through the following equivalence.
    \begin{lemma}
        $D$ has a directed edge odd cycle transversal of size at most $k$ if and only if $D'$ with terminal pairs $\mathcal{T}$ has a symmetric directed multicut of size at most $2k$.
    \end{lemma}
\begin{proof}
    Let $S \subseteq E(D)$ be such that $|S| \leq k$ and $D-S$ is odd-cycle-free. We claim that $S' = \{(u^0, v^1), (u^1, v^0) \mid (u,v) \in S \}$ is a solution for the \dmc instance. It is clearly symmetric and of size $2|S| \leq 2k$, so it remains to show that there are no $t$ to $\phi(t)$ paths in $D' - S'$ for any terminal pair $(t, \phi(t))$. For sake of a contradiction, let $P'$ be such a path. There is a natural way to find a path $P$ corresponding to it in $D$. Simply replace each vertex $v^i$ on $P'$ by $v$ for $P$ and replace each arc $(u^i, v^{i-1})$ by $(u,v)$. Then the first and last vertex of $P$ are the same, and hence $P$ is a cycle. Furthermore, this translation of the path preserves its length, and as $D'$ is bipartite and the two endpoints of $P'$ are in distinct sides of the bipartition, the path $P'$ must have been of odd length. Hence $P$ is an odd cycle in $D$. Finally, observe that $P$ avoids $S$ by our choice of $S'$, which contradicts that $D-S$ is odd-cycle-free.

    For the other direction, let $S' \subseteq E(D')$ be a symmetric set of size at most $2k$ such that $D'-S'$ has no path between terminal pairs. Let $S = \{(u,v) \mid (u^0, v^1) \in S'\}$. Clearly $|S| = \frac{|S'|}{2} \leq \frac{2k}{2} = k$, hence it remains to show that $D-S$ has no odd cycles. For sake of a contradiction, let $C$ be an odd cycle in $D-S$. Then as $M$ is a vertex set hitting all odd cycles, there must be at least one $t \in M$ that hits $C$. We can convert $C$ into a $t^0$ to $t^1$ path $P'$ in $D'-S'$ as follows. Let $t= v_0, v_1, \ldots, v_{\ell-1}, v_\ell = t$ be the vertices visited by $C$ in this order, and note that $\ell$ is odd. The path $P'$ starts at $t^0$ and follows arc $(t^0, v_1^1)$. From here it can follow arc $(v_1^1, v_2^0)$ and repeat this procedure to translate the entirety of $C$. Since $\ell$ is odd, $P'$ has as last vertex $t^1$ and is thus a path between terminal pairs.
    By choice of $S$, this path will be disjoint from $S'$, which contradicts it being a solution.
\end{proof}

To solve the original \deoct instance it thus suffices to decide if the \dmc instance $(D', \mathcal{T}, 2k)$ has a symmetric solution. Our next step will be to reduce the treewidth of the instance. We first argue that the treewidth reduction technique is applicable.

 Let $S \subseteq E(D')$ be a minimal symmetric multicut of size at most $2k$ (if one exists). Let $A \subseteq S$ be a minimal, not necessarily symmetric, multicut contained in $S$. Then $S = A \cup \phi(A)$ as by minimality of $S$ any $A$ needs to pick at least one arc from each symmetric pair of arcs in $S$. By minimality of $A$, every arc in it must be part of some minimal terminal-pair cut. By symmetry of the instance, these two properties must also hold for $\phi(A)$. Hence all arcs in $A$, $\phi(A)$ and as a consequence in $S$ must be part of a minimal $t \to \phi(t)$-cut of size at most $2k$ for some $(t, \phi(t)) \in \mathcal{T}$. It is not difficult to confirm that $D'$ is still an Eulerian digraph, hence all arcs of $S$ are contained in some minimal undirected $t - \phi(t)$ cut of size at most $4k$.

\begin{lemma}
    There is a symmetric set $W \subseteq V(D')$ such that $D^* = \torso(D',W, k)$ has treewidth at most $h(4k,4k)$ and such that $(D', \mathcal{T}, 2k$) has a symmetric solution if and only if $(D^*, \mathcal{T}, 2k)$ has a symmetric solution for a suitable involution $\phi^*$.
\end{lemma}
\begin{proof}
    Let $t_1, t_2, \ldots, t_p$ be the vertices in $M$ for some $p \leq 2k$. We find a set $W_i$ for each $i \in [p]$ by running \Cref{lem:paircut} for each pair of terminals $(t_i^0, t_i^1)$. We set $W = \bigcup_{i=1}^p W_i \cup \phi(W_i)$ and claim that using the treewidth reduction tools we find that $\tw(\torso(D',W,k))$ is appropriately bounded. In particular, by \Cref{lem:paircut} we obtain $\tw(\torso(G', W_i)) \leq g(4k, 4k)$ where $G'$ is the underlying undirected graph of $D'$. By symmetry of the graph the same holds for $\tw(\torso(G, \phi(W_i))$ and thus by \Cref{lem:twunion} we know that $\tw(\torso(G,W)) \leq 2p \cdot g(4k,4k) \leq 4k \cdot g(4k,4k) \leq h(4k,4k)$, and since the underlying undirected graph without parallel edges of $\torso(D',W,k)$ is a subgraph of $\torso(G,W)$, we conclude it has a similar upper bound on its treewidth. 
    By construction of $W$ we obviously have $\phi(W) = W$. This allows us to define a natural symmetry function $\phi^*$ for $D^*$ based on $\phi$.
    
    It remains to show that the two instances are equivalent.
    Let $S$ be a minimal symmetric solution on $D'$ of size at most $2k$. Then by the argument preceding this lemma we know that all endpoints of arcs in $S$ must be in $W$. We claim that $S$ is also a solution on $D^*$, which by construction of $D^*$ remains a symmetric solution. For sake of a contradiction, let $P^*$ be a $t$ to $\phi^*(t)$ path in $D^*-S$ for a terminal pair $(t, \phi^*(t)) \in \mathcal{T}$. By construction of $D^*$, if $P^*$ uses any arc introduced by the torso operation then it can be replaced in $D'$ by the path internally disjoint from $D[W]$ that gave rise to this arc. Replacing each such arc results in a $t$ to $\phi(t)$ path in $D'-S$, which contradicts $S$ being a solution.

    Now let $S$ be a minimal symmetric solution for $D^*$. By construction of $D^*$, $S$ will not contain any arcs introduced by the torso operation, i.e. all arcs of $S$ are arcs in $D'$. As the symmetry of these arcs are maintained for $D'$, it remains to argue there is no terminal pair path in $D' - S$. This follows directly from the fact that any subpath outside of $D'[W]$ that such a path could use can be replaced by an arc between its endpoints in $W$, which exists by definition of the torso, yet is certainly not in $S$.
\end{proof}

At this point we know that $D^*$ has bounded treewidth, but we want to specifically show that $D^*$ has a tree decomposition of bounded width that respects the symmetry of the graph. Let $g : \mathbb{N} \to \mathbb{N}$ be a function such that every undirected graph of treewidth at least $g(p)$ contains a $(p \times p)$-grid (henceforth referred to as a grid of size $p$) as a minor. By Churzhoy and Tan~\cite{DBLP:journals/jctb/ChuzhoyT21} we know that $g(p) = O(p^9 \text{ poly} \log p)$. We prove the following lemma.

\begin{lemma} \label{lem:symtd}
    $D^*$ has a tree decomposition of width at most $2 \cdot g((h'(k))^2)$ for $h'(k) = 1 + h(4k,4k)$ such that for every bag $B$ of the decomposition and every vertex $v \in V(D^*)$, $v \in B$ if and only if $\phi(v) \in B$. 
\end{lemma}
\begin{proof}
    Let $G$ be the graph obtained from $D^*$ and by identifying each vertex pair $(v^0 , v^1)$ to a single vertex $v$. Observe that a tree decomposition of $G$ can be turned into a tree decomposition of $D^*$ with symmetric bags by replacing each such vertex $v$ of $G$ in each bag by the two vertices $v^0, v^1$ of $D^*$ that were identified to form $v$. This only doubles the size of each bag, hence we will prove the claim by bounding the treewidth of $G$.
    For sake of a contradiction, assume that $\tw(G) > g((h'(k))^2)$. 
    By choice of $g$, that means $G$ has as a minor a grid $B$ of size at least $h'(k)^2$. We will show that this implies $D^*$ has as undirected minor a grid of size at least $h'(k) = 1 + h(4k,4k)$ contradicting that $D^*$ has treewidth at most $h(4k,4k)$. 
    
    Consider the minor model $(V_b)_{b \in V(B)}$ of $B$ in $G$. For each $b \in V(B)$, let $T_b$ be a spanning tree of $G[V_b]$.  
    Consider an arc $(u,v)$ that is part of $T_b$. Since $G$ is formed by identifying symmetric vertex pairs of $D^*$, it must mean one of the following arcs exists: $(u^0, v^1)$, $(u^1,v^0)$ $(u^0,v^0)$ or $(u^1,v^1)$. By symmetry of $D^*$, these arcs come in pairs, i.e. the first two are either both in $D^*$ or are both not in $D^*$, and same for the last two. We proceed to replicate $T_b$ in $D^*$. Select an arbitrary $r \in V_b$ as the root. Start from $r^0$ in $D^*$ and consider the arcs of $T_b$ incident to $r$. By the argument above, each such arc must have two counterparts in $D^*$, one with endpoint $r^0$ and one with endpoint $r^1$. Hence we can select for each arc of $T_b$ incident to $r$ a matching arc in $D^*$ incident to $r^0$. One can easily see that repeating this argument will yield us a copy of $T_b$ in $D^*$. Furthermore, if we repeat the process but start now with $r^1$ as root, we find another copy of $T_b$, vertex-disjoint from the first copy. Let $T_b^0$ and $T_b^1$ refer to these two copies of $T_b$ in $D^*$. Find such pairs of trees for each $b \in V(B)$ and let $G^*$ be the (undirected) graph obtained from $D^*$ by contracting all such trees $T_b^0$ and $T_b^1$ to single vertices $b^0$ and $b^1$ and removing all other vertices from the graph.

    At this point we can attempt to replicate the grid minor model in $G^*$.
    Fix some embedding of $B$ on the plane. Then we start at the top-left vertex $b_{00}$ of $B$, for which we have the two candidate vertices $b_{00}^0$ and $b_{00}^1$ in $G^*$. Select $b_{00}^0$ and we will try to discover a copy of $B$ starting from here. 
    We look at a neighbour of $b_{00}$, let us say $b_{01}$. In the original minor model on $G$ there is an arc between $V(T_{b_{00}})$ and $V(T_{b_{01}})$, hence by symmetry there must be an arc from $V(T_{b_{00}}^0)$ to at least one of $V(T_{b_{01}}^0)$ and $V(T_{b_{01}}^1)$. Thus we can find an arc from $b_{00}^0$ to at least one of $b_{01}^0$ or $b_{01}^1$. Add one such arc to our structure. 
    By repeatedly finding such arcs we can step-by-step build a structure resembling the grid $B$, but now in $G^*$. Assume that we build the grid from left to right, top to bottom. There is one potential obstacle to this structure being a proper grid. For vertices $b_{ij}$ of $B$ that are the bottom right corner of a square (a $C_4$) there will be a connection made in our structure both from the vertex above it and to the left of it. However, one arc might connect to $b_{ij}^0$ while the other connects to $b_{ij}^1$ (see \Cref{fig:mismatch} for an illustration). As a result, there will not be a $C_4$ in our structure. Regardless, we can continue the construction procedure by choosing to continue from either only $b_{ij}^0$ or only $b_{ij}^1$ in the next step. We address the raised problem as follows.

    Note that in the above we constructed a structure by starting a growing procedure from $b_{00}^0$. We can in exactly the same way grow a structure starting from $b_{00}^1$. Let $B^0$ and $B^1$ denote these two structures.
    By symmetry of $D^*$, these two structures can be constructed in such a way that they are entirely arc-disjoint. This is achieved by, when in the growing procedure we have a choice between connecting to or continuing from $b_{ij}^0$ and $b_{ij}^1$, making sure both structures make a different choice. 

    Now, in case a mismatch occurs in a square of the grid, then observe that it means that all 8 vertices involved in the square are connected through these 8 arcs that are partitioned over $B^1$ and $B^2$ (i.e. it forms a $C_8$). Let us refer to such mismatched squares of the grid as \emph{unifiers}. 
    If either $B_1$ or $B_2$ (and by symmetry then both) contains a subgrid of 
    size at least $h'(k)$ then we are done. Otherwise, for all subgrids of $B^0$ and $B^1$ of at least this size, there will be a unifier in the corresponding region of $B_1 \cup B_2$. Partition the grid into $h'(k)$ vertex-disjoint subgrids, structured themselves in a grid like structure (see \Cref{fig:gridofgrids}). Each of these subgrids must contain a unifier, and as a consequence it is a connected subgraph of $B_1 \cup B_2$. That means we can treat each subgrid as a branching set, and all branching sets together form a new grid, which is a minor of $D^*$. Since each subgrid is connected, the problem of a mismatch cannot occur. The grid found in this way has size $h'(k)$, concluding the proof. 
\end{proof}

\begin{figure}[t]
  \centering
  \begin{subfigure}{0.48\textwidth}
    \includegraphics{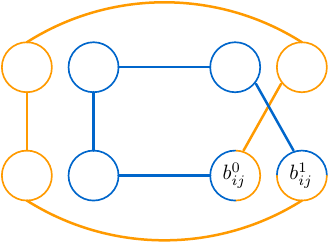}
      \centering
    \caption{Mismatch when constructing a model of the grid. Each vertex $b_{ij}$ on the square of the grid has two candidate vertices $b_{ij}^0$ and $b_{ij}^1$. The structure $B^0$ is in blue and $B^1$ in orange, and both connect to both candidates for the bottom left corner.}
    \label{fig:mismatch}
  \end{subfigure}
  \hfill
  \begin{subfigure}{0.48\textwidth}
    \centering
    \includegraphics{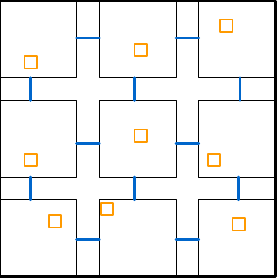}
    \caption{The grid of grids construction. Each subgrid must contain at least one unifier (indicated in orange). 
    }
    \label{fig:gridofgrids}
  \end{subfigure}
  \caption{}
  \label{fig:both}
\end{figure}

The function $h$ in \Cref{thm:twreduction} is such that $h(4k,4k) = 2^{O(k^2)}$, and hence $D^*$ has a symmetric tree decomposition of size at most $2^{O(k^2)}$. 
We can compute a symmetric tree decomposition as stated in \Cref{lem:symtd} by following the setup at the start of its proof. We construct $G$ and compute a tree decomposition of it of width at most $2^{O(k^2)}$ using the algorithm by Korhonen~\cite{DBLP:conf/focs/Korhonen21} in $2^{2^{O(k^2)}} \cdot |V(G)|$ time. Furthermore, we can turn this into a nice tree decomposition of same width and with $2^{O(k^2)} \cdot |V(G)|$ nodes in $2^{O(k^2)} \cdot |V(G)|$ time \cite{DBLP:books/sp/CyganFKLMPPS15}.
We then reverse the identification of all vertices to find a `symmetrically-nice' tree decomposition $(T, \{X_b \mid b \in V(T) \})$ of $D^*$ of width at most $2^{O(k^2)}$. For a node $b \in V(T)$ let $G_b$ be the subgraph of $G$ consisting of all vertex and arcs introduced in the subtree of $T$ rooted at $b$. Let $K_b$ denote the set of all individual terminals that are part of a terminal pair in $\mathcal{T}$ and are in $G_b$.

We then solve the problem of finding a symmetric multicut through dynamic programming on the tree decomposition. The main idea is that the symmetry of our tree decomposition allows us to force only symmetric solutions to be considered, as symmetric arcs can be processed simultaneously. 
We introduce the following two definitions.
\begin{definition}[Reachability relation]
    A \emph{reachability relation} $R$ on a set $U \subseteq V(D)$ of vertices on a digraph $D$ is a reflexive and transitive relation on $U$. The relation is \emph{valid} for a set of terminal pairs $\mathcal{T}$ if $R \cap \mathcal{T} = \emptyset$.
\end{definition}
A digraph $D$ realises a reachability relation $R$ on a set $U \subseteq V(D)$ if $(u,v) \in R$ if and only if there is a (potentially empty) path from $u$ to $v$ in $D$.

\begin{definition}
   For every node $b \in V(T)$, every valid reachability relation $R$ on $K_b \cup X_b $ and every symmetric set $\tilde{S} \subseteq E(G_b[X_b])$, we define $c[b,R,\tilde{S}]$ to be the minimum size of a set $S \subseteq E(G_b)$ such that 
\begin{enumerate}
    \item $S$ is symmetric,
    \item $S \cap E(G_b[X_b]) = \tilde{S}$, and
    \item $G_b-S$ realises reachability relation $R$.
\end{enumerate}
If no such set $S$ exists we put $c[b,R,\tilde{S}] = \infty$.
\end{definition}
Let $r$ be the root of $T$. Then we can find the size of a minimum symmetric DMC solution for $(D^*, \mathcal{T}, k)$ by taking $\min_R c[r,R,\emptyset]$ where $R$ ranges over all valid reachability sets on the terminals. We will show in a series of lemmas how to compute the values of $c$ for each type of node appearing in $T$, after which they can be computed for the whole tree in a bottom-up manner.

\begin{lemma}[Leaf node] \label{lem:leafnode}
    Let $b$ be a leaf node, then $c[b, \emptyset, \emptyset] = 0$. 
\end{lemma}
\begin{proof}
    For a leaf node $b$ the graph $G_b$ is empty and thus there is only one valid choice. 
\end{proof}

\begin{lemma}[Introduce vertex node] \label{lem:intvertnode}
     Let $b$ be an introduce vertex node with a child $b'$ such that $X_b = X_{b'} \cup \{v^0, v^1\}$ for some $v^0, v^1 \notin X_{b'}$. Then
     \begin{align*}
        c[b,R,\tilde{S}] =
            \begin{cases}
              c[b',R \setminus \{(v^0, v^0), (v^1, v^1)\},\tilde{S}] & \text{if $v^0$ and $v^1$ only appear reflexively in $R$}\\
              \infty & \text{otherwise}
            \end{cases}   
\end{align*}
\end{lemma}
\begin{proof}
    The two introduced vertices $v^0$ and $v^1$ are isolated in $G_b$ as none of their incident arcs have been introduced yet. Hence it is impossible for any solution to realise any reachability relation that involves $v^0$ or $v^1$ (except for their reflexive relations $(v^0, v^0)$ and $(v^1, v^1)$). Therefore, if $R$ contains $v^0$ or $v^1$ in any other pair, $c[b,R,\tilde{S}]$ must be $\infty$. If they indeed appear only in reflexive pairs, then any solution $S$ on $G_b$ realises reachability relation $R$ if and only if $S$ realises $R \setminus \{(v^0, v^0), (v^1, v^1)\}$ on $G_{b'}$.
\end{proof}

\begin{lemma}[Introduce edge node] \label{lem:intedgenode}
    Let $b$ be an introduce edge node with a child $b'$. Let $(u,v)$ and $\phi^*((u,v))$ be the arcs introduced by $b$. Then 
    \begin{align*}
  c[b,R,\tilde{S}] =
    \begin{cases}
      2 + c[b',R,\tilde{S} \setminus \{(u,v), \varphi^*((u,v))\}] & \text{if $(u,v), \varphi^*((u,v)) \in \tilde{S}$}\\
      \min_{R'} c[b',R',\tilde{S}] & \text{otherwise}
    \end{cases}       
\end{align*}
Where $R'$ in the second case ranges over all valid reachability relations such that the transitive closure of $R' \cup \{ (u,v), \varphi^*((u,v)) \}$ is equal to $R$.
\end{lemma}
\begin{proof}
    First consider the case where $(u,v), \varphi^*((u,v)) \in \tilde{S}$. Then note that these two arcs are in any potential solution $S$ and hence $G_b - S = G_{b'} - (S \setminus \{(u,v), \varphi^*((u,v))\})$. Furthermore any solution for $c[b',R,\tilde{S} \setminus \{(u,v), \varphi^*((u,v))\}]$ can be turned into a solution for $c[b,R,\tilde{S}]$ by adding these two arcs to $S$, from which  equality in the first case follows.  

    Next consider the second case. Note that a solution $S$ for $c[b,R,\tilde{S}]$ realises some reachability relation $R'$ on $G_b'$ that can be obtained from $R$ by potentially removing some pairs. Hence $S$ will be a valid solution for some $R'$ and thus $c[b,R,\tilde{S}] \geq \min_{R'} c[b',R',\tilde{S}]$. For the other direction, note that any solution $S'$ for $c[b', R', \tilde{S}]$ by definition of $R'$ realises reachability $R$ on $G_b$, hence $c[b,R,\tilde{S}] \leq \min_{R'} c[b',R',\tilde{S}]$.
\end{proof}

\begin{lemma}[Forget node] 
    Let $b$ be a forget node with a child $b'$ such that $X_b = X_{b'} \setminus \{v^0, v^1\}$ for some $v^0, v^1 \in X_{b'}$. Then
    \begin{equation*}
        c[b, R, \tilde{S}] = \min_{\tilde{S}', R'} c[b', R', \tilde{S}'] 
    \end{equation*}
    Where $\tilde{S}'$ ranges over all supersets of $\tilde{S}$ obtained by adding a symmetric set of arcs of $E(G_{b'}[X_{b'}])$ incident to $v^0$ and $v^1$, and $R'$ ranges over all valid reachability relations of $K_{b'} \cup X_{b'}$ that is equal to $R$ when restricted to $K_{b} \cup X_{b}$.
\end{lemma}
\begin{proof}
   First note that $G_b = G_{b'}$. Let $S$ be a solution for $c[b,R,\tilde{S}]$. Then $S$ is also a solution for $c[b', R_S, S\cap E(G_{b'}[X_{b'}]) ]$ where $R_s$ is the reachability relation on $K_{b'}\cup X_{b'}$ realised by $G_{b'} - S$. As these are both considered for $R'$ and $\tilde{S}'$, we have $c[b, R, \tilde{S}] \geq \min_{\tilde{S}', R'} c[b', R', \tilde{S}'] $. For the other direction, consider a solution $S'$ for some $c[b', R', \tilde{S}']$. Then $G_b - S$ realises $R$ and $S \cap E(G_b[X_b]) = \tilde{S}$ so it is a candidate for $c[b,R,\tilde{S}]$, hence  $c[b, R, \tilde{S}] \leq \min_{\tilde{S}', R'} c[b', R', \tilde{S}'] $.
\end{proof}

\begin{lemma}[Join node] \label{lem:joinnode}
    Let $b$ be a join node with children $b_1$ and $b_2$. Then
\[c[b,R,\tilde{S}] = \min_{R_1, R_2} c[b_1, R_1, \tilde{S}] + c[b_2, R_2, \tilde{S}] - |\tilde{S}|\]
Where $R_1$ and $R_2$ range over the valid reachability relations for both parts, such that the transitive closure of $R_1 \cup R_2$ is equal to $R$.
\end{lemma}
\begin{proof}
    Remember  that $X_b = X_{b_1} = X_{b_2}$ in this case. Let $S$ be an arc set of minimum size fitting the definition of $c[b,R,\tilde{S}]$. Let $S_1 = S \cap E(G_{b_1})$ and $S_2 = S \cap E(G_{b_2})$. Let $R_1'$ and $R_2'$ be the reachability relations realised respectively by $S_1$ in $G_{b_1}$ and $S_2$ in $G_{b_2}$. 
    Then as $R_1' \cup R_2'$ has $R$ as transitive closure, we get the following inequality. 
    \begin{align*}
        c[b,R,\tilde{S}] &= |S| = |S_1| + |S_2| - |\tilde{S}| \\
        &\geq c[b_1, R_1, \tilde{S}] + c[b_2, R_2, \tilde{S}] - |\tilde{S}| \\
        &\geq   \min_{R_1, R_2} c[b_1, R_1, \tilde{S}] + c[b_2, R_2, \tilde{S}] - |\tilde{S}|
    \end{align*}

    Conversely, let $S_1$ and $S_2$ now be the minimum sets according to the definitions of $c[b_1, R_1, \tilde{S}]$ and $ c[b_2, R_2, \tilde{S}]$ for some $R_1$ and $R_2$ that have as transitive closure of their union $R$. Then it suffices to observe that $S_1 \cup S_2$ is a set of arcs of $G_b$ realising reachability relation $R$ and with $(S_1 \cup S_2) \cap E(G_B[X_B]) = \tilde{S}$, and hence $c[b,R, \tilde{S}] \leq |S_1 \cup S_2| = |S_1| + |S_2| - |\tilde{S}|$. As this holds for each such $S_1$ and $S_2$, the second direction of the inequality and hence the lemma follows. 
\end{proof}

Finally, note that since the tree decomposition has width bounded in $2^{O(k^2)}$, and since we have at most $4k$ terminal vertices in $D^*$, the dynamic programming algorithm runs in time $2^{2^{O(k^2)}} \cdot |V(D^*)|$,  proving \Cref{thm:deoct}.

\subsection{Proof of the Treewidth-reduction Statement for Edge-cuts} \label{sec:twrproofs} 

\begin{lemma}[Edge version of Lemma 2.4 of \cite{DBLP:journals/talg/MarxOR13}] \label{lem:ebase}
    Let $s,t$ be two vertices in graph $G$ such that the size of a minimum $s$-$t$ cut is $\ell>0$. Then there is a collection $\mathcal{X} = \{X_1,\ldots , X_q\}$ of sets where $\{s\} \subseteq X_i \subseteq V(G) \setminus \{t\}$ for $1 \leq i \leq q$, such that
    \begin{enumerate}
        \item $X_1 \subset X_2 \subset \cdots \subset X_q$, \label{it:lem24-1}
        \item $|\delta(X_i)| = \ell$ for every $1 \leq i \leq q$, and \label{it:lem24-2}
        \item Every $s$-$t$ cut of size $\ell$ is fully contained in $\bigcup_{i=1}^{q} \delta(X_i)$ \label{it:lem24-3}

    \end{enumerate}
    
\end{lemma}
\begin{proof}
    We construct a directed graph $D$ on the vertex set of $G$ by adding the two arcs $(v,u)$ and $(u,v)$, both of capacity one, for every edge $\{u,v\}$ of $G$. The following claim follows directly from this construction.
    \begin{claim}
        For every $R \subseteq V(G) \setminus \{t\}$ that contains $s$, the edge-set $\delta_G(R)$ is an $s$-$t$ cut of size $k$ in $G$ if and only if $\delta_D^+(R)$ is an $s\rightarrow t$ cut of weight $k$ in $D$ (where $\delta_D^+(R)$ is the set of all outgoing arcs of $R$).
    \end{claim}
    Since every minimal $s \rightarrow t$ cut is equal to $\delta_D^+(R)$ for some $R$, it suffices to find a collection $\mathcal{X}$ as in the lemma, yet now on $D$ for directed cuts $s \rightarrow t$ cuts instead.

    We run $\ell$ rounds of the Ford-Fulkerson algorithm on $D$ to find a maximum $s\rightarrow t$ flow, let $D'$ be the residual graph obtained based on this flow. Let $C_1, \ldots, C_p$ be a topological ordering of the strongly connected components of $D'$. Since $\ell > 0$ there is a flow from $s$ to $t$ so there will also be a $t$ to $s$ path in $D'$. Since this is a maximum flow, there will be no $s$ to $t$ path in $D'$. So, if we assume $t \in C_x$ and $s \in C_y$, then we must have $x < y$. For $x < i \leq y$, we define $X_i := \bigcup_{j=i}^p C_j$. 
    For $1 \leq i \leq y - x$, we define  $X_i := \bigcup_{j=y+1-i}^p C_j$.

    Property \ref{it:lem24-1} is immediate from this definition of $\mathcal{X}$.
    For property \ref{it:lem24-2}, we need to confirm that $|\delta_D^+(X_i)| = \ell$ for each $i$. We know that $X_i$ is a suffix of the topological ordering of the strongly connected components of $D'$, and hence has no outgoing arcs in $D'$. So all arcs of $\delta_D^+(X_i)$ must have been turned into incoming arcs in $D'$, which happens if this arc is used by the flow. Since no flow enters $X_i$ (it would give rise to an outgoing arc in $D'$), there can be at most $\ell$ arcs in $\delta_D^+(X_i)$. Furthermore note that $\delta_D^+(X_i)$ is an $s\rightarrow t$ cut and thus must have size at least $\ell$, hence $|\delta_D^+(X_i)| = \ell$.

    To prove that property \ref{it:lem24-3} is satisfied, let $F$ be a minimum $s \rightarrow t$ cut in $D$. Let $R$ be the set of nodes reachable from $s$ in $D-F$, which implies that $F = \delta_D^+(R)$. Consider an arc $(a,b) \in F$. Since this arc is in a minimum cut and since a maximum flow saturates all minimum cuts, we know that the arc $(b,a)$ will be present in the residual graph $D'$. Furthermore, since there is a flow going through $(a,b)$, there must also be a flow from $s$ to $a$ and from $b$ to $t$, and thus an $a$ to $s$ path and a $t$ to $b$ path in $D'$. Hence, if $a \in C_{i(a)}$ and $b \in C_{i(b)}$, then $x \leq i(b) \leq i(a) \leq y$. It remains to prove that $a$ and $b$ are in different strongly connected components of $D'$ (i.e. $i(a) \neq i(b)$), from which it follows that $(a,b) \in X_{i(a)}$. Assume for a contradiction that $a$ and $b$ are in the same strongly connected component. Since $a \in R$ and $b \notin R$, this must mean there is an arc $(c,d)$ going out of $R$ in $D'$ (so $c \in R$ and $d \notin R$). Note that such an arc is part of $F$ and thus part of a minimum cut, which means it is saturated by the flow and hence not in $D'$, a contradiction. 
\end{proof}

For a graph $G$ and a vertex set $U \subseteq V(G)$, we denote by $\partial(U)$  the \emph{boundary of $U$}, which consists of all vertices in $U$ that have a neighbour outside of $U$. Let $\ell$ be the minimum size of an $s$-$t$ cut, then we say that a $s$-$t$ cut $C$ has \emph{excess} $e$ if $|C| = \ell + e$. 

\begin{definition}[Torso]
    Let $G$ be an undirected graph and $U \subseteq V(G)$. The graph $\torso(G,U)$ has vertex set $U$ and for each $u,v \in U$ an edge $\{u,v\}$ if there is a path from $u$ to $v$ in $G$ that is internally vertex-disjoint from $U$. 
\end{definition}

\begin{proposition} \label{prop:torsosg}
    For a graph $G$ and sets $C,U \subseteq V(G)$, $\torso(G[U], C\cap U)$ is a subgraph of $\torso(G,C)$.
\end{proposition}

We will make use of the following lemmas regarding the treewidth of torsos.
\begin{lemma}[{\cite[Lemma 2.8]{DBLP:journals/talg/MarxOR13}}] \label{lem:twunion}
Let $G$ be a graph and $C_1, \ldots, C_r$ be subsets of $V(G)$ and let $C := \bigcup_{i=1}^rC_i$. We have $\tw(\torso(G,C)) \leq \sum_{i=1}^r \tw(\torso(G,C_i))$.    
\end{lemma}

\begin{lemma}[{\cite[Lemma 2.9]{DBLP:journals/talg/MarxOR13}}] \label{lem:twr2.9}
    Let $C' \subseteq V(G)$ be a set of vertices and let $R_1, \ldots, R_r$ be components of $G \setminus C'$. For every $1 \leq i \leq r$, let $C'_i \subseteq R_i$ be subsets and let $C'' := C' \cup \bigcup_{i=1}^r C_i'$. Then we have
    \[\tw(\torso(G,C'')) \leq \tw(\torso(G,C')) + \max_{i=1}^r \tw(\torso(G[R_i], C_i')) + 1.\]
\end{lemma}
 
\begin{lemma}[Edge version of Lemma 2.11 of \cite{DBLP:journals/talg/MarxOR13}] \label{lem:paircut}
    Let $s,t$ be two vertices of graph $G$ and let $\ell$ be the minimum size of an $s$-$t$ cut. For some $e \geq 0$, let $I$ be the union of all minimal $s$-$t$ cuts having excess at most $e$ (i.e. size at most $k = \ell + e$). Then there is an $f(\ell, e) \cdot (|E(G)| +|V(G)|)$ time algorithm that returns a set $W$ that contains the endpoints of all edges in $I$ such that $\tw(\torso(G,W)) \leq g(\ell,e)$ (for some $f(\ell, e)$ and $g(\ell, e)$ such that both are $2^{O(e\ell)}$).  
\end{lemma}
\begin{proof}
    We prove the lemma by induction on the excess $e$. Start by invoking \Cref{lem:ebase} to get a collection $\mathcal{X} = \{X_1, \ldots, X_q\}$. We define $S_i := N(X_i) \cup \partial(X_i)$ for $1 \leq i \leq q$. Additionally set $X_0 = \emptyset$, $X_{q+1} = V(G)\setminus\{t\}$, $S_0 = \{s\}$ and $S_{q+1} := \{t\}$.
    Then for $1 \leq i \leq q+1$ let $L_i := X_i \setminus (X_{i-1} \cup S_{i-1} \cup S_i)$.
    Observe that $|S_i| \leq 2\ell$ and $S_i$ each is a $s$-$t$ separator of $G$.

    Now for  $1 \leq i \leq q+1$ and disjoint non-empty $A,B \subseteq S_{i-1} \cup S_i$ we define $G_{i,A,B}$ as the graph $G[L_i \cup A \cup B]$ with new vertices $a$ and $b$ and $k+1$ subdivided edges from $a$ to every vertex in $A$, and similar for $b$ and $B$. This ensures that these edges are never part of a minimal $a-b$ cut of size at most $k$, and hence all such cuts are in fact $A-B$ cuts. 

    \begin{claim} \label{claim:inductivemincut}
        If an edge $c = \{x,y\}$ with $y \in L_i$ is in $I$, then there are disjoint non-empty $A,B \subseteq S_{i-1} \cup S_i$ such that $c$ is part of a minimal $a-b$ cut $C_2$ in $G_{i,A,B}$ of size at most $k$ and excess at most $e-1$.
    \end{claim}
    \begin{claimproof}
        As $c \in I$, there is a minimal $s$-$t$ cut $C$ of size at most $k$ that contains $c$. Let $C_2 := \{ \{a,b\} \in C \mid \{a,b\} \cap L_i \neq \emptyset\}$ (i.e. the edges of $C$ with at least one endpoint in $L_i$) and $C_1 := C \setminus C_2$. Partition $S_{i-1} \cup S_i$ into sets $A$ and $B$ with vertices reachable from $s$ in $G-C$ and with those not reachable from $s$ in $G-C$, respectively. 
        
        We first argue that $A$ and $B$ are non-empty. As $C$ is a minimal cut, there must be an $s$-$t$ path $P$ in $G$ that intersect $C$ exclusively in $c$. Note that $S_{i-1}$ separates $s$ from $y$ and $S_i$ separates $y$ from $t$. Thus the prefix from the path $P$ up till $x$ must intersect at least one $u \in S_{i-1}$, and since this prefix is entirely disjoint from $C$ there is a $s-u$ path in $G-C$, ergo $u \in A$. Next consider a suffix of $P$ starting from $y$, which gives a $y-t$ path, disjoint from $C$, which must intersect $S_i$. Thus there is a $w \in S_i$ for which there is a $w-t$ path in $G-C$. If $w$ was reachable from $s$ in $G-C$ this would then imply that there is a $s$-$t$ path in $G-C$ which contradicts $C$ being a $s$-$t$ cut. Hence it must be the case that $w$ is not reachable from $s$ in $G-C$ and thus $w \in B$.

        Next, we show that $K_2$ is an $a-b$ cut in $G_{i,A,B}$. For sake of a contradiction, let $P$ be an $a-b$ path in $G_{i,A,B} - C_2$. By construction of $G_{i,A,B}$ this path contains a minimal $A-B$ subpath, meaning a subpath of which the interior is contained in $L_i$. This subpath must clearly exist in $G-C_2$ and thus by definition of $C_2$ also in $G-C$. However, by definition of $A$ there exists an $s-A$ path in $G-C$, which could consequently be extended with this $A-B$ subpath, which results in an $s-B$ path in $G-C$, contradicting the definition of $B$. Hence $C_2$ must be an $a-b$ cut.

        We continue to prove that $C_2$ is a \emph{minimal} $a-b$ cut in $G_{i,A,B}$.
        We again use a proof by contradiction. Assume that $C_2 \setminus \{\{p,q\}\}$ is still an $a-b$ cut for some $\{p,q\} \in C_2$. Since $C$ is a minimal $s$-$t$ cut of $G$, there must exist a $s$-$t$ path $P$ in $G$ of which only intersections $C$ in $\{p,q\}$. Assume without loss of generality that $q \in L_i$. Then for similar reasons as above $P$ has to intersect both $A$ and $B$. Let $P'$ be a minimal $A-B$ subpath of $P$. Then by construction all internal vertices of $P'$ belong to $L_i$,
        hence $P'$ is a $A-B$ path in $G_{i,A,B} - (C_2\setminus\{\{p,q\}\}$, which can trivially be extended to an $a-b$ path, contradicting that $C_2 \setminus \{\{p,q\}\}$ is an $a-b$ cut.

        It remains to show that $C_2$ has excess at most $e-1$ (with regards to being an $a-b$ cut in $G_{i,A,B}$). Let $C_2'$ be a minimum $a-b$ cut in $G_{i,A,B}$. We claim that $C' := C_1 \cup C_2'$ is an $s$-$t$ cut in $G$. If this is not the case, then there must be a $s$-$t$ path $P$ in $G - C'$. This path must intersect $C_2 \setminus C_2'$ and thus must visit $L_i$. We must then again have that $P$ visits both $A$ and $B$. Take a minimal $A-B$ subpath $P'$ of $P$, which by construction is an $A-B$ path in $G_{i,A,B} - C_2'$ and can be trivially extended to an $a-b$ path. This contradicts $C_2'$ being a $a-b$ cut, hence $C'$ must be an $s$-$t$ cut. Since $C'$ contains an edge with an endpoint in $L_i$ it cannot be of minimum size, so $|C'| > \ell$. Putting this together gives us that
        $|C_2| - |C_2'| = (|C_1| + |C_2|) - (|C_1| + |C_2'|) = |C| - |C'| < k - \ell = e$. 
        This concludes the proof of the claim.
    \end{claimproof}

    Having proven the claim, we can now present our algorithm for computing such a set $C'$. Remember that we use induction on the excess $e$.   
    First, we define $W_0 := \bigcup_{i=1}^q S_i$. In case $e=0$ we set $W = W_0$ and we are done. Note that we can define a tree decomposition of $\torso(G,W_0)$ of width at most $4\ell$ by taking as bags the sets $S_i \cup S_{i+1}$ for $1 \leq i < q$. By \Cref{lem:ebase} this set contains all endpoints of minimum $s$-$t$ cuts in $G$. 
    If $e>0$, then for every $1 \leq i \leq q+1$ and subsets $A,B$ of $S_{i-1} \cup S_i$, we use the induction hypothesis to compute a set $W_{i,A,B}$ containing all minimal $a-b$ cuts of excess at most $e-1$ in $G_{i,A,B}$. We set $W = W_0 \cup \bigcup W'_{i,A,B}$ where $W'_{i,A,B} := W_{i,A,B} \setminus (A \cup B \cup \{a,b\})$ (note that $A,B \subseteq C_0$ already). By \Cref{claim:inductivemincut}, $W$ contains the endpoints of all edges in minimal $s$-$t$ cuts of size at most $k$ and excess at most $e$. 

    Next we want to bound the treewidth of $\torso(G,W)$. Observe that by the induction hypothesis $\tw(\torso(G_{i,A,B},W_{i,A,B})) \leq g(\ell,e-1)$ and by construction  $W'_{i,A,B} \subseteq L_i$. 
    Observe that \\ $\torso(G_{i,A,B}[L_i], W_{i,A,B} \cap L_i) =  \torso(G[L_i], W'_{i,A,B})$ is a subgraph of $\torso(G_{i,A,B}, W_{i,A,B})$ by \Cref{prop:torsosg}, and thus $\tw(\torso(G[L_i], W'_{i,A,B})) \leq \tw(\torso(G_{i,A,B}, W_{i,A,B}))$. Consider the set $W'_i := \bigcup W'_{i,A,B}$. Then by using \Cref{lem:twunion}, noting  that for any $i$ we find at most $3^{4\ell}$ sets $A,B$, we get $\tw(\torso(G[L_i], W'_i)) \leq 3^{4\ell} \cdot g(\ell, e-1)$. Finally, by \Cref{lem:twr2.9} we can conclude that $\tw(\torso(G,W)) \leq \tw(\torso(G, W_0)) + \max_i \tw(\torso(G[L_i], W'_i)) + 1 \leq 4\ell + 3^{4\ell} \cdot g(\ell, e-1) + 1$. This means we can choose $g$ such that $g(\ell ,e) = 2^{O(e\ell)}$. 

    The running time remains essentially unchanged compared to the original proof, and hence we refer to \cite{DBLP:journals/talg/MarxOR13} for this.
\end{proof}

\thmtwreduction*

\begin{proof}
    Consider each distinct pair $s,t \in T$ and compute a set $W_{s,t}$ containing all minimal $s$-$t$ cuts of size at most $k$ using \Cref{lem:paircut} and adding $s$ and $t$ to it (in case there is no $s$-$t$ cut of size at most $k$, $W_{s,t} = \{s,t\}$). Note that adding $s$ and $t$ to $W_{s,t}$ increases the treewidth of the torso induced by it by at most two. Let $W$ be the union of all these sets $W_{s,t}$, then the torso of $W$ in $G$ has by \Cref{lem:twunion} treewidth bounded by $|T|^2 2^{O(k^2)}$.
\end{proof}

\section{Reducing CMSO model-checking on Eulerian digraphs to arc-unbreakable instances}\label{sec:reduce-to-unbreakable}

The goal of this section is to prove \Cref{thm:cmsofpt}. 

\thmcmsofpt

We follow the general approach of Lokshtanov, Ramanujan, Saurabh and Zehavi~\cite{DBLP:conf/icalp/LokshtanovR0Z18} where possible. That being said, in \cite{DBLP:conf/icalp/LokshtanovR0Z18} the concept of \emph{boundaried graphs} plays a fundamental role. Roughly speaking, this is a graph where a set of vertices is marked as the boundary, mainly with the purpose of being able to `glue' two compatible boundaried graphs to each other.
However, in our setting of Eulerian digraphs and arc-cuts, it is more suitable to consider a set of arcs to be the boundary instead. One fundamental result on traditional boundaried graphs is that when considering those graphs with a fixed, finite boundary size, there is only a finite amount of different behaviours they can exhibit with regards to a fixed \cmso property.
The goal of the first half of this section is to prove that the same holds for arc-boundaried graphs.  
We achieve this by showing that the arc-boundaried setting can be reduced to the vertex-boundaried setting, which allows us to lift this finite state property to arc-boundaried graphs. Specifically, we prove this for (arc-boundaried) structures, which are a generalisation of digraphs useful when modelling \cmso properties. 

After showing that these structures have finite state, we can prove \Cref{thm:cmsofpt}. This is done by recursion on the size of the structure. If the structure is small (with regards to $\tilde{s}(k)$) then it is trivially unbreakable. Otherwise, a large yet breakable instance can be partitioned into two large parts separated by a small cut. We can take one of these parts in isolation, regard it as an arc-boundaried structure, and recursively find for it a small representative (meaning an arc-boundaried structure that behaves identically with regards to our \cmso property), which is guaranteed to exist by the finite state property. We then substitute in the original instance the chosen part by its representative to obtain an equivalent instance of smaller size, on which we can recurse as long as the structure remains breakable. We ensure that the instance remains Eulerian throughout by picking a representative that is internally balanced and has the same imbalance along the boundary as the part of the graph that it replaces.

\subsection{The framework}

\subparagraph{Arc-boundaried Digraphs.}
We start by introducing all the necessary definitions.

\begin{definition}
    A \emph{digraph with half-arcs} $D$ is a tuple $(V(D),E(D),H(D))$ where $(V(D),\allowbreak E(D))$ describes a directed graph and $H(D)$ is a relation on $V(D) \cup \{ \bot \}$  where for each $(x,y) \in H(D)$, exactly one of $x$ or $y$ is in $V(D)$ and the other is equal to $\bot$. The elements of $H(D)$ are called \emph{half-arcs}. 
\end{definition}

In a digraph $D$ with half-arcs, the in-degree $d_D^-(v)$ (out-degree $d_D^+(v)$) of a vertex $v$ is equal to the number of elements of $E(D) \cup H(D)$ for which it is the head (tail). For these graphs, we still say it is \emph{balanced} if $d_D^-(v) = d_D^+(v)$ for all $v \in V(D)$. We assign to each half-arc a type of either $\head$ or $\tail$ based on whether its tail or head is $\bot$, respectively. For a digraph with half-arcs $D$ and $U \subseteq V(D)$ we define the induced subgraph $D[U]$ to have vertex set $U$, as arcs all edges of $D$ with both endpoints in $U$ and as half-arcs those in $H(D)$ incident to $U$, as well as all the arcs in $\delta(A)$ that are turned into half-arcs by setting occurrences of vertices outside of $U$ to $\bot$.

 We continue by introducing the concept of arc-boundaries, which motivates why we are interested in half-arcs.

\begin{definition}
    An \emph{arc-boundaried digraph} is a digraph with half-arcs $D$ together with an injective labelling $\lambda : H(D) \to \mathbb{N}$. We refer to $H(D)$ as the \emph{boundary} and say that the \emph{label set} $L((D, \lambda))$ of the arc-boundaried digraph is the set $\{ \lambda(e) \mid e \in H(D)\}$. 
\end{definition}

When the labelling is clear from the context, we may sometimes refer to the arc-boundaried graph $(D,\lambda)$ as simply $D$. 
We say that two arc-boundaried graphs $(D_1,\lambda_1)$ and $(D_2, \lambda_2)$ are \emph{compatible} if $ L(D_1) = L(D_2)$ and for each $i \in L(D_1)$, the half-arcs $\lambda^{-1}_{1}(i)$ and $\lambda^{-1}_{2}(i)$ have distinct types. For two half-arcs $e_1$ and $e_2$ with distinct types (i.e. one is $\head$ and the other is $\tail$), we define $e_1 \oplus e_2 = (u_1,v_2)$ if $e_1 = (u_1, \bot)$ and $e_2 = (\bot ,v_2)$ and $e_1 \oplus e_2 =  (u_2,v_1)$ if $e_1 = (\bot, v_1)$ and $e_2 = (u_2, \bot)$.

\begin{definition}
    Let $D_1$ and $D_2$ be two compatible arc-boundaried digraphs. Then $D_1 \oplus D_2$ is the digraph (not-boundaried and without half-arcs) obtained by taking the disjoint union of the digraphs of $D_1$ and $D_2$ and for each $i \in L(D_1)$, replacing the half-arcs $\lambda_1^{-1}(i)$ and $\lambda_2^{-1}(i)$ by $\lambda_1^{-1}(i) \oplus \lambda_2^{-1}(i)$. 
\end{definition}

To help with notation, when gluing two arc-boundaried digraphs $D_1$ and $D_2$ we define a mapping $h_{D_1 , D_2}$ that maps each arc or half-arc of $D_1$ or $D_2$ to the arc in $D_1 \oplus D_2$ that it gives rise to. Hence for $e \in E(D_1) \cup E(D_2)$ we get $h_{D_1 , D_2}(e) = e$ and for $i \in L(D_1)$ we get $h_{D_1 , D_2}(\lambda_1^{-1}(i)) = h_{D_1 \oplus D_2}(\lambda_2^{-1}(i)) = \lambda_1^{-1}(i) \oplus \lambda_2^{-1}(i)$.
The following observation will be pivotal for guaranteeing that the reduced instance in \Cref{thm:cmsofpt} is Eulerian. 

\begin{observation} \label{obs:abdbalance}
    Let $D_1$ and $D_2$ be two compatible arc-boundaried digraphs. Then $D_1 \oplus D_2$ is balanced if and only if $D_1$ and $D_2$ are balanced. 
\end{observation}

\subparagraph{Structures and Properties.}
Next, we generalise the notion of arc-boundaried digraphs to \emph{arc-boundaried structures}. First, we introduce the standard notion of a structure, following the exposition in \cite{DBLP:conf/icalp/LokshtanovR0Z18}.

\begin{definition}
    A \emph{structure} $\alpha$ is a tuple with as first element a (directed) graph denoted by $D_\alpha$, and each remaining element is either a subset of $V(D)$, a subset of $E(D)$, a vertex in $V(D)$ or an edge in $E(D)$. The \emph{arity} of the structure is the number of elements in the tuple. 
\end{definition}
Given a structure $\alpha$ of arity $p$ and a $i \in [p]$, we let $\alpha[i]$ refer to $i$'th element in the tuple $\alpha$ and hence $\alpha[i] = D_\alpha$. When $S$ is either a single vertex or edge, or a vertex or edge set, then we can append $S$ to $\alpha$, denoted by $\alpha \diamond S$, to obtain a new structure of arity $p+1$ such that $(\alpha \diamond S)[p+1] = S$ and for all $i \in [p]$ we have $(\alpha \diamond S)[i] = \alpha[i]$. 
We now define the type of a structure.

\newcommand{\graph}{\mathrm{graph}}
\newcommand{\vertex}{\mathrm{vertex}}
\newcommand{\vertexset}{\mathrm{vertex\, set}}
\newcommand{\edge}{\mathrm{edge}}
\newcommand{\edgeset}{\mathrm{edge\, set}}
\newcommand{\ttt}{\mathrm{true}}
\newcommand{\fff}{\mathrm{false}}

\begin{definition}
    Let $\alpha$ be a structure of arity $p$. The \emph{type} of $\alpha$, denoted by $\type(a)$, is a tuple of arity $p$ such that $\type(\alpha)[1] = \graph$ and for every $i \in \{2,3,\ldots,p\}$, define
\[
\type(\alpha)[i]
=
\begin{cases}
    \vertex,
        & \text{if } \alpha[i] \in V(D_\alpha), \\
    \edge,
        & \text{if } \alpha[i] \in E(D_\alpha), \\
    \vertexset,
        & \text{if } \alpha[i] \subseteq V(D_\alpha), \\
    \edgeset,
        & \text{if } \alpha[i] \subseteq  E(D_\alpha)
\end{cases}
\]
\end{definition}

\begin{definition}
    A \emph{property} is a function from the set of all structures to $\{\ttt, \fff \}$. 
\end{definition} 

\emph{Counting Monadic Second Order Logic} (\cmso) provides a syntax to describe properties on graphs. Its syntax includes the standard logical connectives $\neg$, $\wedge$, $\vee$ and $\implies$. It also includes the following atomic formulas. 
\begin{enumerate}
    \item $u \in X$ where $v$ is a vertex (edge) and $V$ is a vertex set (edge set).
    \item $\textbf{tail}(a,v)$ where $a$ is an edge variable and $v$ is a vertex variable. It holds if and only if $v$ is the tail of $a$.
    \item $\textbf{head}(a,v)$ where $a$ is an edge variable and $v$ is a vertex variable. It holds if and only if $v$ is the head of $a$.
    \item $u=v$ for any two variables $u$ and $v$ of the same type.
    \item $\textbf{card}_{q,r}(X) = \textbf{true}$ where $X$ is a set variable and $q$ and $r$ are integers. It holds if and only if $|X| \equiv q  \mod r$
\end{enumerate}
Finally, we can use existential and universal quantifiers over vertices, edges and vertex and edge sets. For example, if $\psi$ is a sentence with a free vertex variable $v$, then $\forall v \in V: \psi(v)$ is a \cmso
sentence that is true if and only if $\psi$ holds when substituting $v$ by each vertex in the graph $D$ to which the formula is applied. We remark that in the context of undirected graphs the predicates for tail and head are usually replaced by an incidence predicate instead. 

A \cmso-sentence $\psi$ can contain several free variables. Each free variable $x$ has a rank $r_x$ associated with it such that $r_x$ is an integer greater than 1. 
Evaluating a \cmso-formula on a structure $\alpha$ involves instantiating the free variables of $\psi$ by the elements in the tuple $\alpha$.
We say that a structure $\alpha$ \emph{matches} a \cmso-formula $\psi$ if the arity of $\alpha$ is at least $\max r_x$ and for each free variable $x$ the type of $x$ in the \cmso-sentence corresponds with $\type(\alpha)[r_x]$. Finally, each \cmso-sentence $\psi$ naturally defines a property $\sigma_\psi$ on structures by mapping structures that do not match it to false, and evaluating $\psi$ on the structure otherwise. \cmsop{$\psi$} describes the decision problem where, given an input structure $\alpha$, one needs to decide whether $\sigma_\psi(\alpha)$ is true.

\subparagraph{Arc-boundaried Structures.}
Now that we have defined the standard notion of structures, we can combine it with our concept of arc-boundaried digraphs. Intuitively, one should consider arc-boundaried structures as a means to allow us to split and glue structures, similar as to how an arc-boundaried digraph allows us to split and glue digraphs. We now give the formal definitions.

\begin{definition}
    An \emph{arc-boundaried structure} $\alpha$ is a tuple with as first element an arc-boundaried digraph denoted by $D_\alpha$, and each remaining element is either a subset of $V(D)$, a subset of $E(D) \cup H(D)$, a vertex in $V(D)$, an element in $E(D) \cup H(D)$, or the symbol $\star$. The \emph{arity} of the arc-boundaried structure is the number of elements in the tuple. 
\end{definition}

\begin{definition}
    Let $\alpha$ be an arc-boundaried structure of arity $p$. The \emph{type} of $\alpha$, denoted by $\type(\alpha)$, is a tuple of arity $p$ such that $\type(\alpha)[1] = \mathrm{arc\text{-}boundaried\, digraph}$ and for every $i \in \{2,3,\ldots,p\}$, define
\[
\type(\alpha)[i]
=
\begin{cases}
    \mathrm{vertex},
        & \text{if } \alpha[i] \in V(D_\alpha), \\
    \mathrm{edge},
        & \text{if } \alpha[i] \in E(D_\alpha) \cup H(D_\alpha), \\
    \mathrm{vertex\, set},
        & \text{if } \alpha[i] \subseteq V(D_\alpha), \\
    \mathrm{edge\, set},
        & \text{if } \alpha[i] \subseteq  E(D_\alpha) \cup H(D_\alpha), \\
    \star,
        & \text{otherwise}.
\end{cases}
\]
\end{definition}

\begin{definition}[Compatibility]
Two arc-boundaried structures $\alpha$ and $\beta$ are \emph{compatible}, denoted by $\alpha \sim_c \beta$, if
\begin{itemize}
    \item $\alpha$ and $\beta$ have the same arity $p$,
    \item $D_\alpha$ and $D_\beta$ are compatible,
    \item for every $i \in \{2, \ldots, p\}$, one of the following conditions holds
\begin{enumerate}
    \item $\type(\alpha)[i] = \type(\beta)[i] \notin \{\star, \mathrm{vertex}\}$;
    \item $\type(\alpha)[i] \in \{\mathrm{vertex}, \mathrm{edge}\}$ and $\type(\beta)[i]=\star$;
    \item $\type(\beta)[i]\in\{\mathrm{vertex},\mathrm{edge}\}$ and $\type(\alpha)[i]=\star$.
\end{enumerate}
\item for every $i \in \{2, \ldots, p\}$ such that both $\alpha[i]$ and $\beta[i]$ are edges, $\alpha[i] \in H(D_\alpha)$, $\beta[i] \in H(D_\beta)$ and $\lambda_\alpha(\alpha[i]) = \lambda_\beta(\beta[i])$
\end{itemize}
\end{definition} 

Compatibility of arc-boundaried structures captures when two such structures can be combined to form a single structure. This gluing operation is now defined.

\begin{definition}[Gluing by $\oplus$]
Given two compatible arc-boundaried structures $\alpha$ and $\beta$ of arity $p$, the operation $\alpha\oplus\beta$ is defined as follows.
\begin{itemize}
\item $\alpha\oplus\beta$ is a structure $\gamma$ of arity $p$.
\item $D_\gamma = D_\alpha \oplus D_\beta$.
\item For every $i \in [2, \ldots p]$:
\begin{itemize}
\item if $\alpha[i]$ and $\beta[i]$ are vertex sets, $\gamma[i]=\alpha[i]\cup\beta[i]$; 
\item if $\alpha[i]$ and $\beta[i]$ are edge sets, $\gamma[i]= \{  h_{D_\alpha, D_\beta}(e) \mid e \in \alpha[i] \cup \beta[i]\} $
\item if $\alpha[i]$ and $\beta[i]$ are edges, $\gamma[i]=\alpha[i] \oplus\beta[i]$; 
\item if $\alpha[i]=\star$, $\gamma[i]=\beta[i]$ if $\beta[i]$ is a vertex and $\gamma[i]=h_{D_\alpha, D_\beta}(\beta[i])$ if it is an edge
\item if $\beta[i]=\star$, $\gamma[i]=\alpha[i]$ if $\alpha[i]$ is a vertex and $\gamma[i]=h_{D_\alpha, D_\beta}(\alpha[i])$ if it is an edge
\end{itemize}
\end{itemize}
\end{definition}
We say that an (arc-boundaried) structure $\alpha$ is Eulerian if $D_\alpha$ is balanced. Let $p \in \mathbb{N}$ and $I \subseteq \mathbb{N}$, we then denote with $\mathcal{A}^p_{\subseteq I}$ the set of all arc-boundaried structures of arity $p$ and for which the arc-boundaried digraph has as label set a subset of $I$. In a similar vein, we let $\mathcal{B}_{\subseteq I}^p$ be the set of all Eulerian arc-bounded structures of arity $p$ with label set a subset of $I$. 
Finally, similar as \Cref{obs:abdbalance}, we make note of the following.
\begin{observation}
    A structure $\alpha \oplus \beta$ is Eulerian if and only the arc-boundaried structures $\alpha$ and $\beta$ are.
\end{observation}

\subparagraph{Finite state.}
We first define a compatibility equivalence relation capturing when two arc-boundaried structures can be glued to precisely the same set of arc-boundaried structures. Formally, $\alpha \equiv_c \beta$ when $L(D_\alpha) = L(D_\beta)$ and for each $\gamma$, $\alpha \sim_c \gamma \iff \beta \sim_c \gamma$.
We next define an equivalence relation on Eulerian arc-boundaried structures based on a structure property $\sigma$. 

\begin{definition}
    Given a property $\sigma$ for structures, the \emph{Eulerian canonical equivalence relation} $\equiv^E_\sigma$ is an equivalence relation on Eulerian arc-boundaried structures defined as follows. For two Eulerian arc-boundaried structures $\alpha$ and $\beta$, we let $\alpha \equiv^E_\sigma \beta$ if
    \begin{enumerate}
        \item $\alpha \equiv_c\beta$, and
        \item for all Eulerian arc-boundaried structures $\gamma$ that are compatible with $\alpha$ (and thus $\beta$), $\sigma(\alpha \oplus \gamma) = \mathbf{true} \iff \sigma(\beta \oplus \gamma) = \mathbf{true}$. 
    \end{enumerate}
\end{definition}

Next, we prove that when $\sigma$ is a \cmso property and we have fixed and finite $p$ and $I$, this equivalence relation has only a finite amount of equivalence classes on $\mathcal{B}^p_{\subseteq I}$. This same result was proven in the context of general undirected vertex-boundaried structures by Bodlaender et al.~\cite{DBLP:journals/jacm/BodlaenderFLPST16}. We will prove the result for Eulerian arc-boundaried structures by showing that we can translate them to vertex-boundaried structure while preserving equivalence, hence allowing us to lift the result back to arc-boundaried structures. To keep this presentation compact, we refer to \cite{DBLP:journals/jacm/BodlaenderFLPST16} or \cite{DBLP:conf/icalp/LokshtanovR0Z18} for the relevant definitions related to vertex-boundaried structures. To improve readability, we slightly abuse notation and write $\equiv_\psi^E$ rather than $\equiv_{\sigma_\psi}^E$ from here on out.

\begin{theorem} \label{thm:arcfstate}
    For every \cmso-definable property $\psi$, natural number $p \in \mathbb{N}$ and finite $I \subseteq \mathbb{N}$, the equivalence relation $\equiv^E_{\psi}$ on $\mathcal{B}^p_{\subseteq I}$ has a finite number of equivalence classes.
\end{theorem}
\begin{proof}
Fix a \cmso formula $\psi$, a $p \in \mathbb{N}$ and finite $I \subseteq \mathbb{N}$. We can describe a map $g$ that turns any arc-boundaried structure $\alpha \in \mathcal{B}^p_{\subseteq I}$ to a vertex-boundaried structure $\nu$ with arity $p+3$ and the same label set. The undirected graph $G_{g(\alpha)}$ has vertex set $V(D_\alpha) \cup E(D_\alpha) \cup H(D_\alpha)$. It has edges $t_e = \{u,e\}$ and $h_e = \{e,v\}$ for each edge $e = (u,v) \in E(D_\alpha)$. It additionally has an edge $t_e = \{u,e\}$ for each half-arc $e = (u,\bot)$ and edge $h_e = \{e, v\}$ for each half-arc  $(\bot, v)$ in $H(D_\alpha)$. The vertex boundary of $g(\alpha)$ is set to $H(D_\alpha)$ with labelling $\lambda_{g(\alpha)} = \lambda_\alpha$. We translate the elements $\alpha[2], \ldots, \alpha[p]$ to $g(\alpha)$ directly, where typewise an edge in $\alpha$ becomes a vertex in $g(\alpha)$ and an edge set becomes a vertex set. We additionally append three new elements to the end of $g(\alpha)$, namely the vertex sets $\mathbf{Vt} = V(D_\alpha)$, $\mathbf{Ar} = E(D_\alpha) \cup H(D_\alpha)$ and the edge set $\mathbf{T} = \{t_e \mid e \in E(D_\alpha) \vee (e \in H(D_\alpha) \wedge \type(e) = \tail )\}$. 

We consequently translate $\psi$ into a \cmso formula $\psi'$ for undirected structures capturing the same behaviour as $\psi$ did originally. The atomic formulas for head and tail incidence are replaced as follows.
\[
 \textbf{tail}(a,u)\ \rightsquigarrow\ \exists e\in E\ \big(e\in\mathbf{T}\wedge
   \mathbf{inc}(e,a)\wedge\mathbf{inc}(e,u)\big)\]
   \[
 \textbf{head}(a,v)\ \rightsquigarrow\ \exists e \in E\ \big(e\notin\mathbf{T}\wedge
   \mathbf{inc}(e,a)\wedge\mathbf{inc}(e,v)\big)
\]
Furthermore, when in the formula $\psi$ there is a quantification over  the vertex set of the graph, we include that the vertex should be in $\mathbf{Vt}$. Similarly occurrences of the edge set are modified with $\mathbf{Ar}$. For example, $\forall v \in V :\eta$ becomes $\forall v \in V : v \in \mathbf{Vt} \wedge \eta$.
\begin{claim}
    For every arc-boundaried structure $\alpha$ with $H(D_\alpha) = \emptyset$ (i.e. a non-boundaried structure), $\sigma_\psi(\alpha) = \sigma_{\psi'}(g(\alpha))$
\end{claim}
\begin{claimproof}
This follows from the fact that all atomic sentences of $\psi$ get translated into equivalent sentences for $\psi'$ (with regards to the structures they will be evaluated on), and that thus instantiating the free variables of $\psi$ with the elements of $\alpha$ gives the same truth value as instantiating the free variables of $\psi'$ with the elements of $g(\alpha)$.
\end{claimproof}

\begin{claim}
    For two arc-boundaried structures $\alpha$ and $\gamma$ that are compatible, the vertex-boundaried structures $g(\alpha)$ and $g(\gamma)$ are compatible and $g(\alpha \oplus \gamma) = g(\alpha) \oplus g(\gamma)$. 
\end{claim}
\begin{claimproof}
    First we confirm that $g(\alpha)$ and $g(\gamma)$ are indeed compatible (in context of vertex-boundaried structures). They both have arity $p+3$ and for $i \in \{2, \ldots, p\}$ the translation $g$ consistently maps the elements in the structure that are a $\star$ to a $\star$ and a vertex/edge (set) to a vertex (set). As the conditions for type compatibility are stronger for arc-boundaried graphs than for vertex-boundaried graphs, their types must be compatible. Furthermore, as no $g(\alpha)[i]$ or $g(\gamma)[i]$ for $i \in \{2, \ldots, p+3\}$ is an edge, we only need to confirm that when both $g(\alpha)[i]$ and $g(\gamma)[i]$ are a vertex, that 
    $g(\alpha)[i] \in \delta(G_{g(\alpha)})$, $g(\gamma)[i] \in \delta(G_{g(\gamma)})$ and $\lambda_{g(\alpha)}(g(\alpha[i])) =  \lambda_{g(\gamma)}(g(\gamma[i]))$. However, this follows immediately from the way the boundary is defined and that for compatible $\alpha$ and $\gamma$, it must be that both $\alpha[i]$ and $\gamma[i]$ are half-arcs in their respective boundaries with $\lambda_\alpha(\alpha[i]) = \lambda_\gamma(\gamma[i])$. 

    To see that $g(\alpha \oplus \gamma) = g(\alpha) \oplus g(\gamma)$ first observe that $G_{g(\alpha \oplus \gamma)} = G_{g(\alpha)} \oplus G_{g(\gamma)}$ as they are glued precisely along the half-arcs in both cases (although in one they have been transformed to vertices), and the rest of the graph is unaffected by the order of operations. Additionally, for $i \in \{2, \ldots, p\}$, $\alpha[i] = g(\alpha)[i]$ and $\gamma[i] = g(\gamma)[i]$, hence it is irrelevant in which order the gluing operation occurs. It is not difficult to verify that the same holds for the special elements $\mathbf{Vt}$, $\mathbf{Ar}$ and $\mathbf{T}$ as they are sets and the edges are appropriately carried over through the $h_{D_{\alpha}, D_{\gamma}}$ function. 
\end{claimproof}

Now we can argue that $\equiv^E_{\psi}$ has a finite amount of equivalence classes on $\mathcal{B}^p_{\subseteq I}$. It is clear that $\equiv_c$ indeed only has a finite amount for finite $I$ and $p$. Consider $\alpha, \beta \in \mathcal{B}^p_{\subseteq I}$ and an arc-boundaried structure $\gamma$ compatible with them. We observe that $g(\alpha) \equiv_{{\psi'}} g(\beta)$ implies that $\alpha \equiv^E_{\psi} \beta$ through the following use of the claims.
\[\sigma_\psi(\alpha \oplus \gamma) = \sigma_{\psi'}(g(\alpha \oplus \gamma)) = \sigma_{\psi'}(g(\alpha) \oplus g(\gamma)) = \sigma_{\psi'}(g(\beta) \oplus g(\gamma)) = \sigma_{\psi'}(g(\beta \oplus \gamma)) = \sigma_\psi(\beta \oplus \gamma)\]
We obtain from \cite[Lemma 2.2]{DBLP:journals/jacm/BodlaenderFLPST16} that $\equiv_{{\psi'}}$ has a finite number of equivalence classes on vertex-boundaried structures of arity $p+3$ and label set in $I$. Hence $\equiv^E_{\psi}$ has a finite number of equivalence classes on $\mathcal{B}^p_{\subseteq I}$.
\end{proof}

\subparagraph{Edge-breakability.}
Before we can start to prove \Cref{thm:cmsofpt} we need to formally introduce the notion of breakability, and describe how to decide if a graph is or is not breakable.

\begin{definition}
    Let $D$ be a directed graph and let $s$ and $c$ be natural numbers. $D$ is \emph{$(s,c)$-edge-breakable} if there exists a partition $A,B$ of $V(D)$ such that $|\delta(A)| \leq c$ and $|A|, |B| > s$. If no such partition exists, $D$ is \emph{$(s,c)$-edge-unbreakable.}
\end{definition}
 Where $\delta_D(X)$ (written as $\delta(X)$ when it is clear which graph is intended) is defined as the set of edges with one endpoint inside $X$ and one endpoint in $V(D) \setminus X$. As half-arcs can never contribute to the size of a cut, we say that digraph with half-arcs $D$ is $(s,c)$-edge-(un)breakable if and only if the digraph $(V(D),E(D))$ is. Similarly, we say that an (arc-boundaried) structure $\alpha$ is $(s,c)$-edge-(un)breakable if $D_\alpha$ is. As in the above definition the orientation of the edges is irrelevant, we can directly make use of the following result by Chitnis, Cygan, Hajiaghayi, Pilipczuk and Pilipczuk~\cite{DBLP:journals/siamcomp/ChitnisCHPP16}. 
 \begin{lemma}[{\cite[Lemma 2.2]{DBLP:journals/siamcomp/ChitnisCHPP16}}] \label{lem:findbreakable-con}
     There exists a deterministic algorithm that, given an (un)directed graph $G$ on $n$ vertices, along with integers $c$ and $s$, in time $O(2^{O(\min(s,c) \log (s+c))} \allowbreak n^3 \log n)$ either finds a partition $A,B$ witnessing that $G$ is $(s,c)$-edge-unbreakable or reports that no such partition exists for which $A$ and $B$ both contain a connected component of size greater than $s$. 
 \end{lemma}

 One limitation is that \Cref{lem:findbreakable-con} can only decide breakability with the additional constraint that both sides of the partition must contain a large connected component. The next lemma shows that we can avoid this constraint if we are satisfied with merely an `approximate' breakability-witnessing partition. 

 \begin{lemma} \label{lem:findbreakable}
       There exists a deterministic algorithm that, given an (un)directed graph $G$ on $n$ vertices, along with integers $c$ and $s$, in time $O(2^{O(\min(s,c) \log (s+c))} n^3 \log n)$ either finds a partition $A,B$ witnessing that $G$ is $(\frac{s}{2c},c)$-edge-breakable or reports that $G$ is $(s,c)$-edge-unbreakable. 
 \end{lemma}
 \begin{proof}

    We can easily determine whether $G$ is $(s,0)$-edge-unbreakable by solving a polynomial-time solvable instance of the knapsack problem on the connected components of $G$.
    Next, if $n \leq 2s$ then it is trivially $(s,c)$-edge-unbreakable.
    Otherwise,  let $C_1,C_2, \ldots, C_p$ be the connected components of $G$, ordered decreasingly by their number of vertices. Consider the case where $|(V(C_1)| \leq \frac{s}{2c}$. Then we can always find a $q \in [p]$ such that $\frac{s}{2c} < \sum_{i=1}^q |V(C_i)| \leq \frac{s}{c}$, and since $n > 2s \geq \frac{3s}{2c}$ we get that $\bigcup_{i=1}^q V(C_i)$ and $\bigcup_{i=q+1}^p V(C_i)$ forms a partition witnessing $(\frac{s}{2c},0)$-edge-unbreakability, so we output this for $A$ and $B$. 
    Now consider the second case where $|V(C_1)| > \frac{s}{2c}$. Let $U =  \bigcup_{i=2}^p V(C_i)$. If $|U| > \frac{s}{2c}$ then we again find a trivial partition witnessing $(\frac{s}{2c},0)$-edge-unbreakability, namely $V(C_1)$ and $U$, thus assume henceforth that $|U| \leq \frac{s}{2c}$. We use the algorithm of \Cref{lem:findbreakable-con} on $G[C_1]$ to examine $(\frac{s}{2c}, c)$-edge-breakability. If it finds no partition we report that $G$ is $(s,c)$-edge-unbreakable. Alternatively, it finds a partition $A_1, B_1$ of $V(C_1)$ such that $|A_1|, |B_1| > \frac{s}{2c}$ and $|\delta(A_1)| \leq c$, in which case we output $A_1 \cup U$ and $B_1$. 

    It should be obvious that whenever the algorithm outputs a partition, then it is correct. It remains to argue that $G$ is only $(s,c)$-edge-unbreakable when the algorithm reports so. To this end, assume that $G$ is $(s,c)$-edge-breakable and let $A,B$ be a partition of $V(G)$ witnessing this. We will show that both $G[A]$ and $G[B]$ contain a connected component of size greater than $\frac{s}{2c}$, after which correctness immediately follows by \Cref{lem:findbreakable-con}. Remember, we only enter the final case (and risk falsely reporting unbreakability) if $|U| \leq \frac{s}{2c}$. Split the connected components of $G[A]$ into $\mathcal{C}$ which contains all those components $C$ such that $\delta_G(C) = \emptyset$ (ergo $C$ is a connected component already in $G$) and $\mathcal{C'}$ which contains all remaining components. Thus $\mathcal{C} \subseteq \{ C_2, C_3, \ldots, C_p\}$ and $|\bigcup_{C \in \mathcal{C}} V(C)| \leq |U|$. As each connected component in $\mathcal{C}'$ contains an endpoint of at least one edge in $\delta(A)$, we get $|\mathcal{C}'| \leq c$. Combining these facts with $|A| > s$ tells us that at least one $C' \in \mathcal{C}'$ has $|V(C')| > \frac{s-|U|}{c} \geq \frac{s}{c} - |U| \geq \frac{s}{c} - \frac{s}{2c} = \frac{s}{2c}$. The same argument holds on the $B$ side, which concludes the correctness proof.

    For the running time of the algorithm, observe that the $(s,0)$-edge-unbreakable case can be solved in $O(n^2)$ time with a dynamic programming algorithm. 
    The algorithm from \Cref{lem:findbreakable-con} runs in the claimed time and is only called once, and all other steps can be easily performed in $O(n^2)$ time.   
 \end{proof}

\subparagraph{Reduction to Unbreakable Instances.}
Our goal is to prove an Eulerian analogue to Theorem 3 of \cite{DBLP:conf/icalp/LokshtanovR0Z18}. Observe that while this original theorem could be used to reduce the task of solving the problem on general Eulerian structures to the problem of solving it on unbreakable structures, it would require us to solve it on unbreakable \emph{non-Eulerian} structures, which defeats the purpose of looking at Eulerian structures in the first place.
Our main contribution is in showing that it suffices to solve the problem on unbreakable \emph{Eulerian} structures. 

Before stating the theorem, we need to introduce the notion of the parametrized problem \minecmso{$\psi$} for a fixed \cmso formula $\psi$. It takes as input a structure $\alpha$ and a parameter $k$ such for each $S \subseteq E(D_\alpha)$, the structure $\alpha \diamond S$ matches $\psi$. We are tasked to decide whether there is a set $S \subseteq E(D_\alpha)$ of size at most $k$ such that $\alpha \diamond S$ satisfies $\psi$. Similarly, the \maxecmso{$\psi$} and \eqecmso{$\psi$} are similarly defined, now with the task to find an edge set of size at least $k$ or exactly $k$, respectively.

\thmcmsofpt*

We will follow the general approach of \cite{DBLP:conf/icalp/LokshtanovR0Z18} and make modifications where necessary to preserve the Eulerian property and account for arc-boundaried structures. 
We will in fact focus on a non-parameterized result from which \Cref{thm:cmsofpt} follows. 
\begin{restatable}{theorem}{thmcmso} \label{thm:cmso}
    Let $\psi$ be a \cmso formula. For all $c \in \mathbb{N}$ there exists $s \in \mathbb{N}$ such that if there exists an algorithm that solves \cmsop{$\psi$} on $(s,c)$-edge-unbreakable Eulerian structures in time $O(n^d)$ for some $d > 4$, then there exists an algorithm that solves \cmsop{$\psi$} on general Eulerian structures in time $O(n^d)$.
\end{restatable}

We start by showing how \Cref{thm:cmsofpt} follows from \Cref{thm:cmso}.
\begin{proof}[Proof of \Cref{thm:cmsofpt}]
    For $\textsc{x} = \textsc{min-Edge}$ and for a value of $k \in \mathbb{N}$, we define $\psi_k = \exists_{S \subseteq E(D_\alpha)} |S| \leq k \wedge \tilde \psi (S)$ where $\tilde \psi$ used $S$ as free variable. Note that one can easily express $\exists_{S \subseteq E(D)} |S| \leq k$ in \cmso with a formula of length depending on $k$. Similarly, if $x \in \{ \textsc{max-Edge, eq-Edge}\}$ one can similarly define $\tilde \psi_k$ using $|S| \geq k$ and $|S| = k$. Now observe that \cmso{$\tilde \psi$} is equivalent to \cmsop{$\psi_k$}. 
    
    For each $k \in \mathbb{N}$, we let $\tilde s(k)$ be the value of $s$ obtained from \Cref{thm:cmso} with $c = \tilde c(k)$ and $\psi = \psi_k$. Let $b >4$ be fixed such that for each $k$, there is an $f(k)$ such that \xcmso{$\tilde\psi$} on $(\tilde{s}(k), \tilde{c}(k))$-edge-unbreakable Eulerian structures can be solved in $f(k) \cdot n^b$. As \xcmso{$\tilde\psi$} is equivalent to \cmsop{$\psi_k$}, we find that it can be solved in $O(n^b)$ for a fixed $k$ on these unbreakable structures. By \Cref{thm:cmso}, that means \cmsop{$\psi_k$} and hence \xcmso{$\tilde\psi$} can be solved in in $O(n^b)$ in general Eulerian structures. This shows that for every fixed $k \in \mathbb{N}$, \xcmso{$\tilde\psi$} can be solved in $O(n^b)$ with $b$ independent of $k$, and hence it is nonuniformly FPT. 
\end{proof}

Thus it remains to show \Cref{thm:cmso}. First, consider $\psi$ as a fixed \cmso formula with $p$ free variables and a fixed $c \in \mathbb{N}$. We will set $s \in \mathbb{N}$ to a fixed value depending only on $\psi$ and $c$, the exact value of which we will describe later. Thus, to prove the theorem we assume that there exists an algorithm \Unb that solves \cmsop{$\psi$}  on $(s,c)$-edge-unbreakable Eulerian structures in time $O(n^d)$ for some $d > 4$.

We will instead focus on a different problem, namely {\sc Understand-E[$\psi$]}. To define it we need to introduce some notation. Let $I = [2c]$.
Let $\mathcal{E}_{\equiv^E_\psi}[\mathcal{B}^p_{\subseteq I}]$ be the set of equivalence classes of $\equiv^E_\psi$ when restricted to $\mathcal{B}^p_{\subseteq I}$. For an arc-boundaried structure $\alpha \in \mathcal{B}^p_{\subseteq I}$ we denote with $E_\alpha$ the equivalence class containing $\alpha$. For each equivalence class $E_q \in \mathcal{E}_{\equiv^E_\psi}[\mathcal{B}^p_{\subseteq I}]$, we let $\alpha_{E_q}$ be the arc-boundaried structure in $E_q$ such that its encoding is smallest among all elements of $E_q$. We refer to $\alpha_{E_q}$ as the \emph{representative} of $E_q$, and we let $\mathcal{R}_{\equiv^E_\psi} [\mathcal{B}^p_{\subseteq I}]$ be the set of all representatives. By \Cref{thm:arcfstate} there is a fixed $r \in \mathbb{N}$ that depends only on $\psi$ and $c$ such that both $|\mathcal{E}_{\equiv^E_\psi}[\mathcal{B}^p_{\subseteq I}]|$ as well of the maximum size of a representative is upper bounded by $r$. 
Finally, set $s = (2c+1)r$. 

The {\sc Understand-E[$\psi$]} problem is now defined as follows. Given an Eulerian arc-boundaried structure $\alpha \in \mathcal{B}^p_{\subseteq I}$ that matches $\psi$, output an Eulerian arc-boundaried structure $\beta \in \mathcal{R}_{\equiv^E_\psi} [\mathcal{B}^p_{\subseteq I}]$ such that $E_\alpha = E_\beta$ (i.e. return the representative of $E_\alpha$). We start by showing that {\sc Understand-E[$\psi$]} generalises CMSO[$\psi$]. 

\begin{lemma}
    If there exists an algorithm that solves {\sc Understand-E[$\psi$]} on Eulerian arc-boundaried structures in time $O(n^d)$, then there exists an algorithm that solves \cmsop{$\psi$} on Eulerian structures in time $O(n^d)$
\end{lemma}
\begin{proof}
    Let $\alpha$ be a Eulerian structure given as input for \cmsop{$\psi$}. Let $\alpha'$ be the Eulerian arc-boundaried structure obtained by setting $\alpha$ to have an empty boundary. As $\alpha' \in \mathcal{B}^p_{\subseteq I}$, we can call upon the algorithm from the hypothesis to find the representative $\beta \in \mathcal{R}_{\equiv^E_\psi} [\mathcal{B}^p_{\subseteq I}]$ for $E_{\alpha'}$ in $O(n^d)$ time. 

    Now as $\alpha' \sim_c \beta'$, they must both have an empty boundary. Hence by removing this redundant information we can obtain an Eulerian (not boundaried) structure $\beta$. More concretely, let $\gamma$ be the arc-boundaried structure with $\gamma[1]$ as an empty graph and $\gamma[i]$ equal to either $\star$ or $\emptyset$ depending on whether $\alpha[i]$ is a vertex/edge or a set, respectively. Then $\beta = \beta' \oplus \gamma$ and $\alpha = \alpha' \oplus \gamma$. 
    
    By choice of $s$, we know $s > |V(D_\beta)|$ and hence $\beta$ is a trivially $(s,c)$-edge-unbreakable structure. Therefore, we can call upon \Unb to solve $\beta$ in $O(n^d)$ time. As $E_{\alpha'}= E_{\beta'}$ we get $\sigma_\psi(\alpha) = \sigma_\psi(\alpha' \oplus \gamma) = \sigma_\psi(\beta'\oplus \gamma) = \sigma_\psi (\beta)$. 
\end{proof}

Hence we may focus our attention instead on the {\sc Understand-E[$\psi$]} problem. We start by showing how to solve it on unbreakable instances. 

\begin{lemma}
    There exists an algorithm \UndU that solves {\sc Understand-E[$\psi$]} on $(s-r,c)$-edge-unbreakable Eulerian arc-boundaried structures in time $O(n^d)$.  
\end{lemma}
\begin{proof}
    Let $\alpha$ be an $(s-r,c)$-edge-unbreakable Eulerian arc-boundaried structure given as input. Let $\mathcal{C} = \{ \beta \in \mathcal{R}_{\equiv^E_\psi} [\mathcal{B}^p_{\subseteq I}] \mid \alpha \equiv_c\beta\}$ and let $\mathcal{S}= \{ \gamma \in \mathcal{R}_{\equiv^E_\psi} [\mathcal{B}^p_{\subseteq I}] \mid \alpha \sim_c \gamma\}$. As $\alpha$ is $(s-r,c)$-edge-unbreakable and by choice of $r$, for any $\gamma \in \mathcal{S}$ we have $|V(D_\gamma)| \leq r$ and thus $\alpha \oplus \gamma$ is $(s,c)$-edge-unbreakable, as well as Eulerian. Therefore we can use \Unb to compute $\sigma_\psi(\alpha \oplus \gamma)$ for each $\gamma \in \mathcal{S}$.

    On the other hand, for each $\beta \in \mathcal{C}$ also $|V(D_\beta)| \leq r$ hence by choice of $s$ for each $\gamma \in \mathcal{S}$ the structure $\beta \oplus \gamma$ is Eulerian and $(s,c)$-edge-unbreakable. Therefore in the same way as before we can use \Unb to find $\sigma_\psi(\beta \oplus \gamma)$ for each combination of $\beta$ and $\gamma$. At this point, we find the $\beta \in \mathcal{C}$ such that $\sigma_\psi(\alpha \oplus \gamma) = \sigma_\psi(\beta \oplus \gamma)$ for all $\gamma \in \mathcal{S}$, and output this $\beta$.  

    As $\alpha \in \mathcal{B}^p_{\subseteq I}$ there must be at least one candidate $\mathcal{C}$ for $\beta$, namely the representative for $E_\alpha$. It remains to argue that no $\beta' \in \mathcal{C} \setminus \{\beta\}$ also satisfies the condition that $\sigma_\psi(\alpha \oplus \gamma) = \sigma_\psi(\beta' \oplus \gamma)$ for all $\gamma \in \mathcal{S}$. For sake of a contradiction, assume such a $\beta'$ exists. As $\beta, \beta' \in \mathcal{C}$, we immediately conclude that $\beta \equiv_c \beta'$. Now let $\eta$ be an arbitrary Eulerian arc-boundaried structure compatible with $\beta$ and let $\eta'$ be the representative of $E_\eta$ (hence $\eta' \in \mathcal{C})$.
    Then \[\sigma_\psi(\beta \oplus \eta) = \sigma_\psi(\beta \oplus \eta') = \sigma_\psi(\alpha \oplus \eta') = \sigma_\psi(\beta' \oplus \eta') = \sigma_\psi(\beta' \oplus \eta)\]
    which implies that $\beta \equiv^E_\psi \beta'$, a contradiction to the fact that we choose only one representative per equivalence class. Finally, the running time of \UndU is clearly $O(n^d)$ as a consequence of $\mathcal{R}_{\equiv^E_\psi} [\mathcal{B}^p_{\subseteq I}]$ having constant size.     
\end{proof}

We are now ready to prove that we can solve {\sc Understand-E[$\psi$]} on general Eulerian arc-boundaried structures.

\begin{lemma}
    There exists an algorithm \Und that solves {\sc Understand-E[$\psi$]} on general Eulerian arc-boundaried structures in time $O(n^d)$.  
\end{lemma}
\begin{proof}
    Let $\alpha$ be the Eulerian arc-boundaried structure given as input. We can use \Cref{lem:findbreakable} to decide in $O(n^3\log n)$ time (for fixed $s$ and $c$) whether $\alpha$ is $(s-r,c)$-edge-unbreakable, or find a $(\frac{s-r}{2c}, c)$-edge-breakable witnessing partition. If it is unbreakable, then \UndU can be used to solve it in $O(n^d)$ time. Otherwise, let $A,B$ be a partition of $V(D_\alpha)$ with $|A|,|B| > \frac{s-r}{2c} = r$ and $|\delta(A)| \leq c$. Without loss of generality, assume that the number of half-arcs incident to vertices of $B$ is greater than or equal to the number of half-arcs incident to vertices of $A$. We will describe a recursive procedure for \Und, hence we assume that the algorithm returns the correct output for all instances of size strictly smaller than $\alpha$. 

    We now construct an Eulerian arc-boundaried structure $\beta \in \mathcal{B}^p_{\subseteq I}$ based on the $A$ side of the partition. Let the arc-boundaried digraph $D_\beta$ be equal to $D_\alpha[A]$. To help with notation, let $f : \delta(A) \to H(D_\beta)$ be an injective map that maps each edge in $\delta(A)$ to the half-arc that was created from it, with $f(\delta(A)) := \{ f(e) \mid e \in \delta(A)\}$.
    First observe that $D_\beta$ is Eulerian as the balance of each vertex of $A$ is maintained. Second, $|H(D_\beta)| \leq 2c$ as $|\delta(A)| \leq c$ and the $A$ side contained at most half of the original up to $2c$ half-arcs of $D_\alpha$. Set $\lambda_\beta$ such that the labels for all half-arcs that were present in $D_\alpha$ are maintained, and the half-arcs in $f(\delta(A))$ all get mapped to unique integers in $I = [2c]$. Now, for $i \in \{2, \ldots, p\}$ we set $\beta[i]$ in the following way. 
    \begin{itemize}
        \item If $\type(\alpha)[i] = \star$ then $\beta[i] = \star$. 
        \item If $\type(\alpha)[i] = \mathrm{vertex}$ then if $\alpha[i] \in A$, $\beta[i] = \alpha[i]$ and otherwise $\beta[i] = \star$ .
        \item If $\type(\alpha)[i] = \mathrm{edge}$ then if $\alpha[i] \in E(D_\beta) \cup H(D_\beta)$ we set $\beta[i] = \alpha[i]$, if $\alpha[i] \in \delta(A)$ then $\beta[i] = f(\alpha[i])$ and otherwise $\beta[i] = \star$ .
        \item If $\type(\alpha)[i] = \mathrm{vertex\, set}$ then $\beta[i] = \alpha[i] \cap A$
        \item If $\type(\alpha)[i] = \mathrm{edge\, set}$ then $\beta[i] = (\alpha[i] \cap (E(D_\beta) \cup H(D_\beta))) \cup \{f(e) \mid e \in  \beta[i] \cap \delta(A) \}$ (i.e. all edges that were part of the $A$ side, but edges in $\delta(A)$ get correctly replaced by their corresponding half-arcs). 
    \end{itemize}
    Since $\beta \in \mathcal{B}^p_{\subseteq I}$ and $\beta$ is strictly smaller than $\alpha$, we can recursively call \Und to find the Eulerian representative $\beta'$ such that $E_\beta = E_{\beta'}$. 

    We next define the arc-boundaried structure $\gamma \in \mathcal{B}^p_{\subseteq I}$ which will roughly correspond to replacing the $A$ side of $\alpha$ by $\beta'$. More formally, to construct $D_\gamma$ start by taking the disjoint union of $D_{\beta'}$ with only the half-arcs with label not in $\lambda_\beta(f(\delta(A))$ and $D_\alpha[B]$ with only the half-arcs of $ H(D_\alpha)$. 
    For each edge $e = (u,v) \in \delta(A)$, such that $u \in A$, we add an edge from the tail of $\lambda_{\beta'}^{-1}(\lambda_\beta(f(e)))$ to $v$. Alternatively, if $v \in A$, we add an edge from $u$ to the head of $\lambda_{\beta'}^{-1}(\lambda_\beta(f(e)))$. Similar as for a formal gluing operation, we define here a function $h$ that maps each edge or half-arc of $D_{\beta'}$ and $D_\alpha[B]$ to the edge or half-arc of $D_\gamma$ that it naturally corresponds to. 
    We can determine $\lambda_\gamma$ by, for each half-arc $e \in H(D_\gamma)$, setting $\lambda_\gamma(e)$ to $\lambda_\alpha(e)$ or $\lambda_{\beta'}(e)$ depending on whether $e \in H(D_\alpha)$ or $e \in H(D_{\beta'})$, respectively. To finish our definition of $\gamma$, for each $i \in \{2, \ldots, p\}$ we set the following.
    \begin{itemize}
        \item If $\type(\beta')[i] = \star$ then $\gamma[i] = \alpha[i]$.
        \item If $\type(\beta')[i] = \mathrm{vertex}$, then $\gamma[i] = \beta'[i]$.
        \item If $\type(\beta')[i] = \mathrm{edge}$, then $\gamma[i] = h(\beta'[i])$
        \item If $\type(\beta')[i] = \mathrm{vertex\, set}$, then $\gamma[i] = \beta'[i] \cup (\alpha[i] \cap B)$
        \item if $\type(\beta')[i] = \mathrm{edge\, set}$, then 
        $\gamma[i] = \{ h(e) \mid e \in \beta'[i] \cup (\alpha[i] \cap (E(D_\alpha[B]) \cup H(D_\alpha[B])) \}$ (i.e. the edges and half-arcs from $D_\gamma$ that correspond to those from $\beta'[i]$ or from the $B$ side of $\alpha$). 
    \end{itemize}

    To conclude the algorithm, we now solve the instance $\gamma$, which is strictly smaller than $\alpha$ as $|\beta'| \leq r < |\beta|$. Let $\gamma'$ be the output of the algorithm on instance $\gamma$. We continue to prove that $\gamma'$ is indeed the right representative for $\alpha$. It should be clear from the description that \Und only outputs arc-boundaried structures in $\mathcal{R}_{\equiv^E_\psi} [\mathcal{B}^p_{\subseteq I}]$, so we focus on the following.
    
    \begin{claim}
        $\alpha \equiv^E_\psi {\gamma'}$
    \end{claim}
    \begin{claimproof}
        By our assumption that the algorithm is correct on all inputs of size smaller than $\alpha$, we immediately obtain $\beta \equiv^E_\psi {\beta'}$ and $\gamma \equiv^E_\psi {\gamma'}$.   
        Hence the goal becomes $\alpha \equiv^E_\psi \gamma$. We start by arguing that $\alpha \equiv_c \gamma$. First note that $L(D_\alpha) = L(D_\gamma)$ as $H(D_\alpha)$ is partitioned over $D_\alpha[A]$ (and in turn $D_\beta$) and $D_\alpha[B]$, $\beta \equiv^E_\psi {\beta'}$ and when constructing $\gamma$ the corresponding half-arcs with labels in $L(D_\alpha)$ are maintained. For the same reason, the half-arc with a label $i \in I$ in $D_\alpha$ has type $\head$ if and only if the half-arc with label $i$ in $D_\gamma$ has type $\head$, hence $D_\alpha$ and $D_\beta$ are compatible with the same arc-boundaried digraphs. Clearly, $\gamma$ has the same arity $p$ as $\alpha$. For $i \in \{2, \ldots, p\}$, $\type(\alpha)[i] = \type(\gamma)[i]$ because of how $\beta[i]$ is set, $\beta \equiv_c \beta'$ and how $\gamma[i]$ is consequently defined. Hence $\alpha \equiv_c \gamma$.  

        It remains to prove that for an arbitrary Eulerian arc-boundaried structure $\eta$ that is compatible with $\alpha$, we have $\psi(\alpha \oplus \eta) = \true$ if and only if $\psi(\gamma \oplus \eta) = \true$.  The intuitive idea is that to construct $\gamma$ we replaced in $\alpha$ the $A$ side which was captured by $\beta$ by $\beta'$, which are canonically equivalent. Since we can view both $\alpha \oplus \eta$ and $\gamma \oplus \eta$ instead as gluing something to $\beta$ and $\beta'$, we can apply the equivalence of $\beta$ and $\beta'$ to arrive at our goal.
        
        More formally, let $\mu \in \mathcal{B}^p_{\subseteq I}$ be defined as follows. The arc-boundaried digraph $D_\mu$ is equal to $D_{(\alpha \oplus \eta)}[ V(D_{\alpha \oplus \eta}) \setminus A]$. To set $\lambda_\mu$, note that each edge in $\delta_{D_{(\alpha \oplus \eta)}}(A)$ induces an half-arc $e$ in $H(D_\mu)$ but also can be bijectively mapped to a corresponding half-arc $e'$ in $H(D_\beta)$. Thus we let $\lambda_\mu(e) = \lambda_\beta(e')$, from which it follows that $D_\beta \oplus D_\mu = D_\alpha \oplus D_\eta$. The remaining elements of the arc-boundaried structure $\mu$ can naturally be set such that $\beta \oplus \mu = \alpha \oplus \eta$. 
        In fact, since $\gamma$ was constructed by effectively substituting $\beta$ by $\beta'$ it follows that $\beta' \oplus \mu = \gamma \oplus \eta$. This gives us the following
        \[\sigma_\psi(\alpha \oplus \eta) = \sigma_\psi(\beta \oplus \mu) = \sigma_\psi(\beta' \oplus \mu) = \sigma_\psi(\gamma \oplus \eta)\]
        where the second equality follows from $\beta \equiv^E_\psi \beta'$. 
    \end{claimproof}
    
    It remains to show that \Und runs in time $O(n^d)$. 
    We prove this by induction on the size of the instance that there exists a value $x$ such that it runs in $x \cdot n^d$. The algorithm starts by calling the algorithm of \Cref{lem:findbreakable}, which for fixed $s$ and $c$ in time $y \cdot n^3 \log n$ for some fixed $y$. We require from here $x \geq y$.
    In the base case with no recursive calls, \Und simply calls the algorithm of \Cref{lem:findbreakable} which runs in time $O(n^d)$ and hence in $x \cdot n^d$ for an appropriate choice of $x$. We proceed to the inductive step. 
    Given our input instance $\alpha$, \Und will be recursively applied to $\beta$ and $\gamma$. Let $n_\beta = |V(D_\beta)|$ and $n_\gamma = |V(D_\gamma)|$, then by the inductive hypothesis the two calls will run respectively in time $x \cdot (n_\alpha)^d$ and $x \cdot (n_\gamma)^d$. Thus the total runtime will be upper bounded by
    \[ y\cdot  n^3 \log n + x \cdot (( n_\beta)^d + (n_\gamma)^d).\]
    Let $t = (s-r)/ c$ (hence $t \geq 2r)$ and observe that $n_\beta, n_\gamma \geq s \geq t$. Furthermore, $n_\beta + n_\gamma \leq n + r$ by choice of $r$. Define the function $f(z) = z^d + (n_\beta+n_\gamma- z)^d$. For $0 \leq i \leq \frac{n_\beta+n_\gamma}{2}$ we have $f'(i) \leq 0$ and as at least one of $n_\beta$ or $n_\gamma$ is at most $ \frac{n_\beta+n_\gamma}{2}$, we can derive the following.
    \[( n_\beta)^d + (n_\gamma)^d = f(n_\beta) = f(n_\gamma) \leq f(t) = t^d + (n_\alpha+n_\gamma-t)^d \leq t^d + (n+r-t)^d  \]
    Using this we can find the following upper bound on the runtime. 
    
    \begin{align*}
        & &y\cdot  n^3 \log n + x \cdot (( n_\beta)^d + (n_\gamma)^d) \\
        &\leq  &y\cdot  n^3 \log n + x \cdot (t^d + (n+r-t)^d) \\
        &=  & y\cdot  n^3 \log n + x \cdot t^d + x \cdot (n+r-t)  (n+r-t)^{d-1} \\
        &\leq & y\cdot  n^3 \log n + x \cdot t^d + x \cdot (n+r-t)  n^{d-1} \\
        &= &x \cdot n^d + (y\cdot  n^3 \log n +  x \cdot t^d + x\cdot rn^{d-1}) - x \cdot tn^{d-1}
    \end{align*}
    Hence, to prove an upper bound of $x\cdot n^d$ it suffices to prove that 
    \[ x \cdot tn^{d-1} \geq y\cdot  n^3 \log n +  x \cdot t^d + x\cdot rn^{d-1}.\]
    As $t \geq 2r$, we have $\frac{1}{2}x \cdot tn^{d-1} \geq x\cdot rn^{d-1}$ and since $d > 4$, if we ensure $x/4 \geq y$ also $\frac{1}{4} x \cdot tn^{d-1} \geq y\cdot  n^3 \log n$. It remains to show
    \[ \frac{1}{4} x \cdot tn^{d-1} \geq x \cdot t^d, \]
    but this immediately follows since $n > 2(s-r)$ (otherwise the graph is trivially unbreakable), and $t \leq s-r$. 
\end{proof}

\section{Applications of the Eulerian Breakability Framework}\label{sec:breakapplications}
In this section we demonstrate how \Cref{thm:cmsofpt} can be used to find FPT algorithms by applying it to two cut problems on Eulerian digraphs, namely \demcufull and \dscfull. These are two prime examples of problems for which in the undirected setting breakability proved especially useful to find algorithmic results. The minimum bisection\footnote{Given a graph $G$ and integer $k$, decide whether there exists a partition $A,B$ of $V(G)$ such that $| |A| - |B| | \leq 1 $ and there are at most $k$ edges between $A$ and $B$. In the directed setting, we count the number of arcs from $A$ to $B$.} problem is a third major example of such a problem, however one can easily see that the directed version on Eulerian digraphs reduces to the undirected version of the problem, and is thus immediately FPT. 

When designing algorithms on unbreakable instances, it is common to rely on a colour-coding approach. In such cases we will make use of the following lemma that is stated by Chitnis et al.~\cite{DBLP:journals/siamcomp/ChitnisCHPP16} and based on the work by Naor, Schulman and Srinivasan~\cite{naor1995splitters} to be able to present our algorithms in a deterministic manner.

\begin{lemma} \label{lem:derandom}
Given a set $U$ of size $n$, and integers $0 \leq a,b, \leq n$, one can in time $2^{O(min(a,b)\log(a+b))} \allowbreak n \log n$ construct a family $\mathcal{F}$ of at most $2^{O(min(a,b)\log(a+b))} \log n$ subsets of $U$, such that the following holds: for any sets $A,B \subseteq U$, $A \cap B = \emptyset$, $|A| \leq a$, $|B| \leq b$, there exists a set $S \in \mathcal{F}$ with $A \subseteq S$ and $B \cap S = \emptyset$. 
\end{lemma}

\subsection{Directed Edge Multiway Cut-Uncut}
In this section we focus on 
the directed analogue of the {\sc Edge Multiway Cut-Uncut} problem and prove the following. 

\thmdemwcu*

The first step is to show that DEMCU is \cmso-expressible. To this end, we define the following formula with free variables $R$, $T$ and $S$ that should be interpreted according to the problem definition. We will use several obvious shorthands for sake of readability. 
\begin{align*}
    \psi = &
    \textbf{ cluster}(R) \wedge \forall_{s,t \in T} [\textbf{strong-con}(s,t) \iff (s,t) \in R]
\end{align*}
where 
\begin{align*}
    \textbf{cluster}(R) &= \forall_{u,v,w \in T}[ ((u,v) \in R \implies (v,u) \in R) \\
    & \wedge (((u,v) \in R \wedge (v,w) \in R) \implies (u,w) \in R) ] \\
    \textbf{strong-con}(s,t) &= \forall_{A,B \subseteq V(D)} [\textbf{disjoint}(A,B) \wedge ((s \in A \wedge t \in B) \vee (t \in A \wedge s \in B)) \\
    &\implies \exists_{e \in E(D) \setminus (S \cup R)}\ \exists_{ u,v \in V(D)} \ \textbf{tail}(e,u) \wedge \textbf{head}(e,v) \wedge(u \in A \wedge v \in B ]
\end{align*}

\begin{observation}
  $(D,T,R,k)$ is a \yes-instance of DEMCU if and only if $((D',T,R),k)$ is a \yes-instance of {\sc min-Edge-CMSO[$\psi$]} where $D' = (V(D), E(D) \cup R)$. 
\end{observation}

Thus, by \Cref{thm:cmsofpt}, it suffices to solve {\sc min-Edge-CMSO[$\psi$]} on $(q(k), 2k)$-unbreakable structures where $q$ is the function obtained from the theorem.

\begin{lemma} \label{lem:demcubreakable}
    \minecmso{$\psi$} is FPT on $(q(k), 2k)$-unbreakable structures parameterized by $k$. 
\end{lemma}
\begin{proof}
    First note that we can easily check in polynomial time for our structure $(D',T,R)$ whether $R$ is a cluster. Hence, assume this property to be satisfied, which means it remains to determine whether Eulerian, $(q(k), 2k)$-unbreakable, DEMCU instance $(D,T,R,k)$ is a YES-instance. Keep in mind that the unbreakability property applies to $D'$, rather than $D$.

    Let $S \subseteq E(D)$ be a hypothetical minimal witness to it being a YES-instance. 
    Consider a topological ordering of $D-S$ and let $C_0$ be the largest strongly connected component in the ordering. 
    By minimality we must have that $S$ is precisely the set of arcs going backwards between components in this ordering. Since in an Eulerian digraph the number of backward arcs is invariant under cyclic shifts of the ordering \cite[Lemma 4.4.2]{DBLP:books/sp/18/Wahlstrom18}, we may consider without loss of generality a solution $S$ such that $C_0$ is first in this ordering. 
    Let $C_0, C_1, \ldots, C_r$ be the ordering on the strongly connected components of $D-S$. As $S$ contains all backward arcs, for each prefix $\{C_0, \ldots, C_i\}$ of the ordering there are at most $|S|$ arcs from $V(D) \setminus (C_0 \cup \ldots \cup C_i)$ to $(C_0 \cup \ldots \cup C_i)$ and hence by Eulerianness at most $2|S| \leq 2k$ edges in total between these two parts. 
    Furthermore, by nature of $R$, no arcs in it cross strongly connected components for a solution $S$. As $D'$ is $(q(k), 2k)$-unbreakable, this must mean that $r \leq 2q(k)$. 
    Now, if $C_0$ contains at most $q(k)$ vertices, then $|V(D)| \leq 2q(k)^2$ and the problem could be solved through brute-force, hence assume $|C_0| > q(k)$ from which it immediately follows that $|C_1 \cup C_2 \cup \cdots \cup C_r| \leq q(k)$. We can also conclude at this point that if $R$ has more than $q(k)+1$ equivalence classes, then it must be a NO-instance. Keep in mind that, contrary to the undirected version of the same problem, for an optimal solution there can be components not containing any equivalence class.

    We will find $S$ through colour coding. We colour every arc in $D$ either red or black. We consider a colouring \emph{good} when the following properties are satisfied. 
    \begin{enumerate} 
        \item All arcs between different strongly connected components of $D-S$ are red.
        \item An arbitrary spanning tree within each small component $C_1$ to $C_r$ is coloured black.
        \item For each vertex $v$ in $C_0$ that is the endpoint of an arc with as other endpoint a vertex outside of $C_0$, there is a set of $q(k)+1$ vertices in $C_0$ containing $v$ such that a spanning tree of this set is coloured black.
    \end{enumerate}
    In fact, we will use \Cref{lem:derandom} to deterministically compute a family of colourings, such that at least one corresponds to a good colouring (under the assumption that $S$ exists). As the number of arcs that we constrain in a good colouring is bounded in terms of $k$, this computation can be done in FPT time, after which we try to use all sets in the family, of which there are $f(k) \log n$ for some function $f$. 

    We show how to extract $S$ using the colouring by identifying the components $C_1$ to $C_r$.
    Contract each weakly connected components of $D$ minus the red arcs to a single vertex and let $D^*$ be the resulting graph. 
    Let $B$ be the set of vertices of $D^*$ that were formed by contracting components with at least $q(k)+1$ vertices. Note that for a good colouring all vertices in $B$ must be part of $C_0$, and each of $C_1$ to $C_r$ was contracted to single vertices $c_1$ to $c_r$. 
    At this point, we can guess which one of the at most $q(k)$ equivalence classes should be contained in $C_0$ (if any). Then all other equivalence classes should have been entirely contained in a weakly connected component of $D$ with only the black arcs, and any such component contains at most one class. Hence if neither is the case we can discard this guess, and if no guess is valid we discard the colouring. 
    We mark the vertices of $D^*$ that originate from a component that contained an equivalence class. Now consider the weakly connected components of $D^* - B$. 
    We claim that the vertices $c_1$ to $c_r$ are precisely those vertices contained in components of $D^*-B$ containing at least one marked vertex. 
    Since $B$ separates $c_1$ to $c_r$ from the remainder of the graph, any component contains either exclusively vertices within $C_0$ or outside of it. Any component with a marked vertex must then clearly be composed of a subset of the vertices $\{c_1, \ldots, c_r\}$. Furthermore, any component without a marked vertex will be contained in $C_0$ by minimality of $S$. 
    Having identified the set of vertices $\{c_1, c_2, \ldots, c_r\}$ we can extract $\{C_1, C_2, \ldots, C_r\}$ by reversing the contractions. In order to find $S$ (or another solution) it suffices to guess which of the at most $r! \leq q(k)!$ orderings the found components appear in, and verify whether this leads to a valid solution of size at most $k$. As we are guaranteed that at least one colouring and guess is compatible with $S$, we are guaranteed to find a solution if one exists.
\end{proof}

\subsection{Directed Steiner Cut}\label{ssec:dir-steiner-cut}
The \dscfull problem is the directed analogue of the {\sc Steiner Cut} problem, which is known to be FPT when parameterized by the deletion budget for undirected graphs \cite{DBLP:journals/siamcomp/ChitnisCHPP16}. This inspires us to look at the directed version.

We follow similar steps as for \demcufull. We first define a {\sc CMSO}$_2$ formula $\psi$ that is satisfied by a suitable structure if and only if the \dsc instance is a YES-instance. Then we solve the problem on unbreakable structures through derandomised colour-coding. 
\[ \psi = \exists_{V_1, V_2, \cdots, V_p \subseteq V(D)} [ \bigwedge_{i=1}^p \textbf{strong-comp}(V_i) \wedge \exists_t \in T: (t \in V_i) \wedge \forall_{j \neq i}(V_i \neq V_j)]  \] 
where
\begin{align*} 
    \textbf{strong-comp}(V_i) &= \textbf{strong-con}(V_i) \wedge \forall_{U \subseteq V(D)}[V_i \subset U \implies \neg\textbf{strong-con}(U)]    \\
    \textbf{strong-con}(V_i) &= \forall_{A,B \subseteq V_i} [\textbf{disjoint}(A,B) 
    \implies \\
    & \exists_{e \in E(D) \setminus S} \exists_{ u,v \in V(D)} \textbf{tail}(e,u) \wedge \textbf{head}(e,v) \wedge (u \in A \wedge v \in B) \wedge \\
    & \exists_{e \in E(D) \setminus S} \exists_{ u,v \in V(D)} \textbf{tail}(e,u) \wedge \textbf{head}(e,v) \wedge (u \in B \wedge v \in A)] 
\end{align*}

\begin{observation} \label{obs:dsc}
    $(D,T,p,k)$ is a \yes-instance for \dsc if and only if $((D,T,p),k)$ is a \yes-instance for  {\sc min-Edge-CMSO[$\psi$]}.
\end{observation}

\begin{lemma} \label{lem:dscbreak}
    \minecmso{$\psi$} is FPT on $(q(k, p), 2k)$-unbreakable structures parameterized by $k+p$. 
\end{lemma}
\begin{proof}
    Similar as for \demcufull, we start by considering a hypothetical minimal solution $S \subseteq E(D)$. Let $C_0$ be the largest strongly component in $D-S$ and as $D$ is Eulerian assume without loss of generality that the topological ordering of $D-S$ is $C_0, C_1, \ldots, C_r$ for some $r$. Then $S$ contains all backwards arcs of this ordering. Repeating the argument found in \Cref{lem:demcubreakable} tells us that either $|V(D)| \leq 2q(k,p)$ (and the problem can be solved by brute-force), or $|C_1 \cup C_2 \cup \cdots \cup C_r| \leq q(k,p)$. We again colour every arc in $D$ either red or black and consider a colouring good when it satisfies the following properties.
    \begin{enumerate}
        \item All arcs between different strongly connected components of $D-S$ are red.
        \item An arbitrary spanning tree within each small component $C_1$ to $C_p$ is coloured black.
        \item For each vertex $v$ in $C_0$ that is the endpoint of an arc with as other endpoint a vertex outside of $C_0$, there is a set of $q(k,p)+1$ vertices in $C_0$ containing $v$ such that a spanning tree of this set is coloured black.
    \end{enumerate}
    Using \Cref{lem:derandom} we compute in FPT time a family of colourings such that at least one is good, if a good colouring exists (i.e. if a solution $S$ exists).  Let $D^*$ be the graph $D$ after contracting all black arcs (i.e. all weakly connected components of $D$ without the red arcs become single vertices). That means that each small component $C_i$ was contracted to a single vertex $c_i$ for $i \in \{1 \ldots, r\}$. For each vertex of $D^*$ we want to guess if it was one of $\{c_1 \ldots, c_r\}$ or whether it belongs to $C_0$.
    Let $B$ be the set of all vertices of $D^*$ to which at least $q(k,p)+1$ vertices were contracted. Then we know all of $B$ belongs to $C_0$ and in particular $B$ separates the vertices belonging to $C_0$ from $\{c_1, \ldots, c_r\}$. In particular, that means that every weakly connected component of $D^*-B$ must be either a subset of $\{c_1, \ldots, c_r\}$ or be entirely contained in $C_0$. Let $H_i$ denote a connected component of $D^* - B$. We compute both the cost and value of choosing $H$ not to be in $C_0$. If $|V(H)| > q(k,p)$ then by $(q(k,p), 2k)$-breakability of $D$ it must be that $H$ belongs to $C_0$. Otherwise, we set the value simply to the number of vertices of $H$ that are a terminal or to which a terminal $T$ was contracted. We minimise the cost by trying all of the at most $q(k,p)!$ topological orderings of the vertices of $H$ and record the minimum number of backward arcs. We also include all arcs to $B$, as $C_0$ is fixed to be the first component of the topological ordering of $D-S$. This also allows us to compute the cost and value of each weakly connected component independently. After computing all costs and values, we can find a choice of components leading to at least $p$ strongly connected components and at most $k$ cut arcs, if one such choice exists for this guess, by setting up a polynomial-time solvable knapsack instance with these components as items. 
\end{proof}

Combining \Cref{obs:dsc}, \Cref{lem:dscbreak} and \Cref{thm:cmsofpt} will give us the FPT result that we are looking for. In fact, one technically needs a variation on \Cref{thm:cmsofpt} in order to handle a parameterization other than just $k$. However, given \Cref{thm:cmso} it should be clear that such a variation is not hard to prove in an analogues manner as the current proof by defining for a given $p$ the sentence $\psi_p$ as the formulation for $\psi$ as seen above, after which we solve the \minecmso{$\psi_p$} problem.

\begin{theorem}
    \dscfull on Eulerian digraphs is (nonuniformly) FPT parameterized by $k+p$. 
\end{theorem}

The need to include $p$ as parameter comes from the fact that the length of the {\sc CMSO}$_2$ formula, and thus the function for $q$ obtained from \Cref{thm:cmsofpt}, depends on $k$ and $p$. \ For the undirected version of the problem it is observed in \cite{DBLP:journals/siamcomp/ChitnisCHPP16} that one can reduce the instance to get $p$ to be linear in $k$. This reduction crucially relies on the fact that deleting an edge from an undirected graph increases the number of connected components by at most one. In contrast, deleting an arc from a directed graph can increase the number of strongly connected components by up to $n-1$. 

\section{Hardness results} \label{sec:hardness}
In this section we show that the two problems considered in \Cref{sec:breakapplications} are W[1]-hard on general digraphs.
    \subsection{Multiway cut-uncut}
    
\begin{theorem}
    \demcufull on general digraphs is W[1]-hard when parameterized by $k$, even when $R$ has only two equivalence classes.
\end{theorem}
\begin{proof}
    We will reduce from the {\sc Clique} problem. Let $(G,k)$ be an instance for {\sc Clique}, we will construct an instance $(D,T,R,k')$ of DEMCU as follows. First, set $k' = 2k$. We describe step by step how we build $D$, $T$ and $R$. The digraph $D$ has two vertices $a$ and $b$, which are both terminals in different equivalence classes connected through an undeletable arc from $b$ to $a$ (realised as $k'+1$ parallel arcs). 
    For each $i \in [k]$, we construct two paths $P_i$ and $Q_i$ from $a$ to $b$. Specifically, $V(P_i) = \{p_0^i, p_1^i, \ldots, p_n^i\}$ and $V(Q_i) = \{q_1^i, \ldots, q_{n+1}^i\}$. The paths have arcs $(p_j^i, p_{j+1}^i)$ for $0 \leq j < n$, and $(q_{j+1}^i, q_j^i)$ for $1 \leq j \leq n$, and are connected to $a$ and $b$ through undeletable arcs $(a, p_0^i)$, $(a, q_{n+1}^i)$, $(p_n^i, b)$ and $(q_1^i,b)$.
    
    As a consequence of the construction so far, in order to ensure that $a$ and $b$ do not end up in the same strongly connected component, every solution needs to select exactly one arc per path $P_i$ and per path $Q_i$ for $i \in [k]$. The choice of arcs here will simulate our choice of vertices for a clique in $G$. In particular, let us denote $V(G)$ by $\{v_1, v_2, \ldots, v_n\}$. We then want that selecting $v_j$ for the clique is simulated by cutting the two arcs going into $p_j^i$ and $q_j^i$ for precisely one $i \in [k]$. 

    As a first step, we want to `synchronise' the choices for $P_i$ and $Q_i$. We do this as follows. For each $i \in [k]$ and $j \in [n]$, we add a terminal $a_j^i$, which belongs to the same equivalence class as $a$, and undeletable arcs $(p_j^i, a_j^i)$, $(q_{j+1}^i, a_j^i)$ and $(a_j^i, a)$. Similarly, we add a terminal $b_j^i$ in the equivalence class of $b$ and add undeletable arcs $(b_j^i, p_{j-1}^i)$, $(b_j^i, q_j^i)$ and $(b, b_j^i)$. See \Cref{fig:demcu-gdaget} for an illustration. One can now confirm that, when a solution selects an arc $(p_{j-1}^i, p_j^i)$, then in order for $a_j^i$ to be reachable from $a$ the solution must select an arc of $Q_i$ appearing after $q_{j+1}^i$. Simultaneously, in order for $b_j^i$ to be able to reach $b$, the selected arc of $Q_i$ must appear before $q_{j}^i$. That means the only valid choice is the arc $(q_{j+1}^i, q_j^i)$. One can confirm that for any $j \in [n]$, selecting these two arcs ensures that all connectivity requirements on the terminals introduced so far is satisfied. 

    As a next step, we will introduce gadgets to ensure that if in paths $P_i$ and $Q_i$ some vertex $j$ is selected through their arc choice, and if paths $P_i'$, $Q_i'$ for $i \neq i'$ a vertex $j'$ is selected, then we must necessarily have that $\{v_j, v_{j'}\} \in E(G)$. To achieve this we will utilise the fact that when a path pair $P_i$ and $Q_i$ make a choice for a vertex $v_j$ through their arc-cuts, then there is one unique index, namely $j$, such that both $p_j^i$ and $q_j^i$ are not reachable from $a$. So, for each distinct pair $i, i' \in [k]$ and for each $j, j' \in [n]$ such that $\{v_j, v_{j'}\} \notin E(G)$, we introduce a terminal vertex $e_{j,j'}^{i,i'}$ in the same class as $a$, and we add an undeletable arc $(e_{j,j'}^{i,i'}, a)$, as well as an undeletable arc from $p_j^i$, $q_j^i$, $p_{j'}^{i'}$ and $q_{j'}^{i'}$ to $e_{j,j'}^{i,i'}$. This ensures that $a$ can reach $e_{j,j'}^{i,i'}$ if and only if in $P_i$ and $Q_i$ the arcs for $v_j$ are not selected, or in $P_{i'}$ and $Q_{i'}$ the arcs for $v_{j'}$ are not selected. 

    Correctness of the reduction now follows from correctness of the individual gadgets, and the observation that these gadgets do not interfere.
\end{proof}

\begin{figure}
    \centering
    \includegraphics{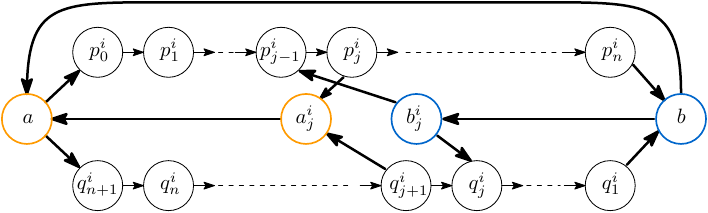}
    \caption{Choice-synchronisation gadget between a path $P_i$ and a path $Q_i$ in the construction of the \demcu instance. Undeletable arcs are drawn in bold. }
    \label{fig:demcu-gdaget}
\end{figure}

\subsection{Directed Steiner cut}

\begin{theorem}
    \dscfull on general digraphs is W[1]-hard when parameterized by $k+p$.
\end{theorem}
\begin{proof}
    We show W[1]-hardness by reducing from the {\sc Clique} problem. Given an instance $(G,k)$ of {\sc Clique}, we construct an instance $(D,T,k,p)$ of \dsc as follows. We set $p = {k \choose 2} + 1$. Construct $D$ by taking as vertex set $V(D) = \{r\} \cup \{a_v \mid v \in V(G) \} \cup \{t_e \mid e \in E(G)\}$. We have arcs $(r, a_v)$ for each $v \in V(G)$. Furthermore, we add an \emph{undeletable} arc $(a_v, r)$ (undeletable arcs can be modelled by $k+1$ parallel arcs). Next, for each edge $uv \in E(G)$, we add undeletable arcs from $a_u$ and $a_v$ to $t_{uv}$. Finally, we add an undeletable arc from each $t_e$ to $r$. We conclude by setting $T = \{r\} \cup \{t_e \mid e \in E(G)\} $. 

    We now show that the reduction is correct. Let $C \subseteq V(G)$ be a clique of size $k$ in $G$. Then we let $X = \{(r,a_v) \mid v \in C\}$ be our cutset for the DSC instance. Clearly $|X| = |C| = k$, hence it remains to show that $D-X$ has at least $p$ strongly connected components with a terminal. For this, it suffices to observe that a vertex $t_{uv}$ is in a strongly connected component distinct from the one containing $r$ if and only if $X$ contains both the arcs from $r$ to $a_u$ and to $a_v$. Thus, the number of strongly connected components of $D-X$ containing a terminal is precisely one plus the number of edges in $G[C]$, which is $k \choose 2$ as $C$ is a clique. 

    For the converse, let $X \subseteq E(D)$ be a cutting set of size most $k$ such that $D-X$ has at least $p$ strongly connected components containing a terminal. Let $C = \{v \in V(G) \mid (r, a_v) \in X\}$, and observe that $|C| = |X| \leq k$. As stated above, we know that a terminal $t_{uv}$ is in its own strongly connected component if and only if both $(r, a_u)$ and $(r, a_v)$ are in $X$. That means that by cutting at most $k$ arcs, at most $k \choose 2$ of such terminals can have become a strongly connected component. Since $p = {k \choose 2} + 1$, we know that we must have that $|X| = k$ and for each pair $(r, a_u), (r, a_v) \in X$, there must exist a terminal $t_{uv}$. By construction, this means that for each pair $u,v \in C$, there must be an edge $\{u,v\} \in E(G)$, ergo $C$ is a clique of size $k$. 
\end{proof}

\bibliographystyle{alpha}
\bibliography{references}

\end{document}